\documentclass[english]{lipics-v2019}
\usepackage[T1]{fontenc}
\usepackage[utf8]{inputenc}
\usepackage{babel}
\usepackage{verbatim}
\usepackage{amsmath}
\usepackage{amsthm}
\usepackage{amssymb}
\usepackage{microtype}
\usepackage{multirow}

\makeatletter
\newtheorem{thm}{\protect\theoremname}
\newtheorem{lem}[thm]{\protect\lemmaname}

\providecommand{\pgfsyspdfmark}[3]{}

\nolinenumbers
\hideLIPIcs

\usepackage{xcolor}

\usepackage{xspace}

\usepackage[textsize=tiny,textwidth=1.5cm]{todonotes}

\renewcommand{\lg}{\log}
\renewcommand{\epsilon}{\varepsilon}
\newcommand{\OoEps}{\frac{1}{\varepsilon}}
\newcommand{\polylg}{\operatorname{polylog}}
\newcommand{\poly}{\mathrm{poly}}
\newcommand{\polylog}{\polylg}
\newcommand{\updateweightedintervals}{(1/\varepsilon)^{O(1/\varepsilon)}\lg^{2}n \lg^{5}N \lg W}

\usepackage{babel}
\providecommand{\corollaryname}{Corollary}
\providecommand{\lemmaname}{Lemma}
\providecommand{\theoremname}{Theorem}

\DeclareRobustCommand{\alg}{%
	\ifmmode
		\mathcal{A}
	\else
		\text{$\mathcal{A}$}\xspace
	\fi
}

\global\long\def\ch{\mathrm{ch}}%
\global\long\def\DD{\mathcal{D}}%
\global\long\def\D{D}%
\global\long\def\S{\mathcal{S}}%
\global\long\def\T{\mathcal{T}}%
\global\long\def\C{\mathcal{C}}%
\global\long\def\L{\mathcal{L}}%
\global\long\def\G{\mathcal{G}}%
\global\long\def\K{\mathcal{K}}%
\global\long\def\R{\mathbb{R}}%
\global\long\def\Z{\mathbb{Z}}%
\global\long\def\N{\mathbb{N}}%
\global\long\def\OPT{\mathrm{OPT}}%
\global\long\def\ALG{\mathrm{DP}}%
\global\long\def\SOL{\mathrm{SOL}}%
\global\long\def\SUB{\mathrm{SUB}}%
\global\long\def\off{\mathrm{off}}%

\newcommand{\Tfine}{\T_{\text{fine}}}
\newcommand{\Zfine}{Z_{\text{fine}}}
\newcommand{\sparse}{\text{sparse}}
\newcommand{\dense}{\text{dense}}

\makeatother

\providecommand{\lemmaname}{Lemma}
\providecommand{\theoremname}{Theorem}

\title{Dynamic Approximate Maximum Independent Set of Intervals, Hypercubes and Hyperrectangles}
\titlerunning{Dynamic Independent Set of Intervals, Hypercubes and Hyperrectangles}

\relatedversion{The conference version of this paper will appear at Symposium on
	Computational Geometry (SoCG) 2020.}

\author{Monika Henzinger}{Faculty of Computer Science, University of Vienna, Vienna, Austria}{monika.henzinger@univie.ac.at}{https://orcid.org/0000-0002-5008-6530}{The research leading to these results has received funding from the European Research Council under the European Community's Seventh Framework Programme (FP7/2007-2013) / ERC grant agreement No.~340506.}
\author{Stefan Neumann}{Faculty of Computer Science, University of Vienna, Vienna, Austria}{stefan.neumann@univie.ac.at}{}{Part of this work was done while visiting Brown University. Stefan Neumann gratefully acknowledges the financial support from the Doctoral Programme ``Vienna Graduate School on Computational Optimization'' which is funded by the Austrian Science Fund (FWF, project no.~W1260-N35). The research leading to these results has received funding from the European Research Council under the European Community's Seventh Framework Programme (FP7/2007-2013) / ERC grant agreement No.~340506.}
\author{Andreas Wiese}{Department of Industrial Engineering, Universidad de Chile, Santiago, Chile}{awiese@dii.uchile.cl}{}{Andreas Wiese was supported by the grant Fondecyt Regular 1170223.}
\authorrunning{M.~Henzinger, S.~Neumann and A.~Wiese}

\acknowledgements{We are grateful to the anonymous reviewers for their helpful
	comments.}

\Copyright{Monika Henzinger, Stefan Neumann and Andreas Wiese}



\begin{document}

\maketitle

\begin{abstract}
Independent set is a fundamental problem in combinatorial optimization.  While
in general graphs the problem is essentially inapproximable, for many important
graph classes there are approximation algorithms known in the offline setting.
These graph classes include interval graphs and geometric intersection graphs,
where vertices correspond to intervals/geometric objects and an edge indicates
that the two corresponding objects intersect.  

We present \emph{dynamic} approximation algorithms for independent set of
intervals, hypercubes and hyperrectangles in $d$ dimensions.  They work in the
fully dynamic model where each update inserts or deletes a geometric object.
All our algorithms are deterministic and have worst-case update times that
are polylogarithmic for constant $d$ and $\epsilon>0$,
assuming that the coordinates of all input objects are in $[0, N]^d$
and each of their edges has length at least 1. We obtain the following results:
\begin{itemize}
	\item For weighted intervals, we maintain a $(1+\epsilon)$-approximate
		solution.
	\item For $d$-dimensional hypercubes we maintain a $(1+\epsilon)2^{d}$-approximate
		solution in the unweighted case and a $O(2^{d})$-approximate solution
		in the weighted case. Also, we show that for maintaining an unweighted
		$(1+\epsilon)$-approximate solution one needs polynomial update
		time for $d\ge2$ if the ETH holds.
	\item For weighted $d$-dimensional hyperrectangles we present a dynamic
		algorithm with approximation ratio $(1+\epsilon)\log^{d-1}N$.
\end{itemize}
\end{abstract}

\section{Introduction}

A fundamental problem in combinatorial optimization is the independent
set (IS) problem. Given an undirected graph $G=(V,E)$ with $n$ vertices
and $m$ edges, the goal is to select a set of nodes $V'\subseteq V$
of maximum cardinality such that no two vertices $u,v\in V'$ are
connected by an edge in $E$. In general graphs, IS cannot be approximated
within a factor of $n^{1-\epsilon}$ for any $\epsilon>0$, unless
$\mathsf{P=NP}$~\cite{zuck07}. However, there are many approximation
algorithms known for special cases of IS where much better approximation
ratios are possible or the problem is even polynomial-time solvable.
These cases include interval graphs and, more generally, geometric
intersection graphs.

In interval graphs each vertex corresponds to an interval on the real line and
there is an edge between two vertices if their corresponding intervals
intersect. Thus, an IS corresponds to a set of non-intersecting intervals on the
real line; the optimal solution can be computed in time
$O(n+m)$~\cite{frank75some} when the input is presented as an interval graph and
in time $O(n \lg n)$~\cite[Chapter~6.1]{kleinberg06algorithm} when the intervals
themselves form the input (but not their corresponding graph). Both algorithms work even in the weighted
case where each interval has a weight and the objective is to maximize the total
weight of the selected intervals.

When generalizing this problem to higher dimensions, the input consists
of axis-parallel $d$-dimensional hypercubes or hyperrectangles and
the goal is to find a set of non-intersecting hypercubes or hyperrectangles
of maximum cardinality or weight. This is equivalent to solving
IS in the \emph{geometric intersection graph} of these objects which
has one (weighted) vertex for each input object and two vertices are
adjacent if their corresponding objects intersect. This problem is
$\mathsf{NP}$-hard already for unweighted unit squares~\cite{FOWLER1981},
but if all input objects are weighted hypercubes then it admits a
PTAS for any constant dimension~$d$~\cite{chan03polynomial,erlebach2005polynomial}.
For hyperrectangles there is a $O((\log n)^{d-2}\log\log n)$-approximation
algorithm in the unweighted case~\cite{ChalermsookC09} and a $O((\log n)^{d-1}/\log\log n)$-approximation
algorithm in the weighted case~\cite{Chan2012,ChalermsookC09}. IS
of (hyper-)cubes and (hyper-)rectangles has many applications, e.g.,
in map labelling~\cite{agarwal98label,map-labelling-ESA}, chip manufacturing~\cite{hochbaum1985approximation},
or data mining~\cite{khanna1998approximating}. Therefore, approximation
algorithms for these problems have been extensively studied, e.g.,
\cite{Chan2012,ChalermsookC09,agarwal98label,JACM-ISSP,chan2004note,Chuzhoy2016}.

All previously mentioned algorithms work in the static offline setting.
However, it is a natural question to study IS in the \emph{dynamic} setting,
i.e., where (hyper-)rectangles appear or disappear, and
one seeks to maintain a good IS while spending only little time after
each change of the graph. The algorithms above are not suitable for
this purpose since they are based on dynamic programs in which $n^{\Omega(1/\epsilon)}$
many sub-solutions might change after an update or they solve linear programs
for the entire input. %
For general graphs, there are several results for maintaining a \emph{maximal}
IS dynamically~\cite{AssadiOSS18,assadi2019fully,DBLP:journals/corr/abs-1804-01823,DBLP:journals/corr/abs-1804-08908,monemizadeh19dynamic,chechik19fully,behnezhad19fully},
i.e., a set $V'\subseteq V$ such that $V'\cup\{v\}$ is not an IS
for any $v\in V\setminus V'$. However, these algorithms do not imply
good approximation ratios for the geometric setting we study: Already
in unweighted interval graphs, a maximal IS can be by a factor $\Omega(n)$
smaller than the maximum IS.
For dynamic IS of intervals, Gavruskin et al.~\cite{gavruskin15dynamic} showed
how to maintain an exact maximum IS with polylogarithmic update time in the
special case when no interval is fully contained in any another interval.

\textbf{Our contributions.} In this paper, we present dynamic
algorithms that maintain an approximate IS in the geometric intersection
graph for three different types of geometric objects: intervals,
hypercubes and hyperrectangles. We assume throughout the paper that
the given objects are axis-parallel and contained in the space $[0,N]^{d}$,
that we are given the value $N$ in advance, and that each edge of
an input object has length at least 1 and at most~$N$. We study the fully dynamic
setting where in each update an input object is inserted or deleted.
Note that this corresponds to inserting and deleting \emph{vertices}
of the corresponding intersection graph. In particular, when a vertex
is inserted/deleted then potentially $\Omega(n)$ edges might be
inserted/deleted, i.e.,
there might be more edge changes per operation than in the standard
dynamic graph model in which each update can only insert or delete
a single edge.

(1)~For independent set in weighted interval graphs we present a dynamic
$(1+\epsilon)$-approximation algorithm. For weighted $d$-dimensional hypercubes
our dynamic algorithm maintains a $(4+\epsilon)2^{d}$-approximate solution; in
the case of unweighted $d$-dimensional hypercubes we obtain an approximation
ratio of $(1+\epsilon)2^{d}$.  Thus, for constant $d$ we achieve a constant
approximation ratio. Furthermore, for weighted $d$-dimensional hyperrectangles
we obtain a dynamic algorithm with approximation ratio of
$(1+\varepsilon)\lg^{d-1}N$.

Our algorithms are deterministic with worst-case update times that
are polylogarithmic in $n$, $N$, and $W$, where $W$ is the maximum
weight of any interval or hypercube, for constant $d$ and $\varepsilon$;
we also show how to obtain faster update times using randomized algorithms
that compute good solutions with high probability. In each studied
setting our algorithms can return the computed IS $I$ in time $O_{d,\varepsilon}(|I|\cdot\mathrm{poly}(\log n,\log N))$,
where $|I|$ denotes the cardinality of~$I$ and the $O_{d,\varepsilon}(\cdot)$
notation hides
factors which only depend on $d$ and $\varepsilon$. \emph{Up to a $(1+\epsilon)$-factor
our approximation ratios match those of the best known near-linear
time offline approximation algorithms for the respective cases} (with
ratios of $2^{d}$ and $O(2^{d})$ via greedy algorithms for unweighted
and weighted hypercubes and $\lg^{d-1}N$ for hyperrectangles~\cite{agarwal98label}).
See Table~\ref{table:results} for a summary of our algorithms.

(2)~Apart from the comparison with the static algorithm we show two lower
bounds: We prove that one cannot maintain a $(1+\epsilon)$-approximate
IS of unweighted hypercubes in $d\ge2$ dimensions with update time
$n^{O((1/\epsilon)^{1-\delta})}$ for any $\delta>0$ (so even with
polynomial instead of polylogarithmic update time), unless the Exponential
Time Hypothesis fails. Also, we show that maintaining a maximum weight
IS in an interval graph requires $\Omega(\log N/\lg\lg N)$ amortized
update time.

\begin{table}
	\begin{centering}
	\begin{tabular}{|l|c|c|}
	\cline{2-3} \cline{3-3} 
	\multicolumn{1}{l|}{} & \multirow{1}{*}{Approximation} & Worst-case update time\tabularnewline
	\multicolumn{1}{l|}{} & \multirow{1}{*}{ratio} & \tabularnewline
	\hline 
	Unweighted intervals  & $1+\epsilon$  & $O_{\epsilon}(1)\log^{2}n\log^{2}N$\tabularnewline
	\cline{2-3}
	 & $1$  & $\Omega(\lg N/\lg\lg N)$\tabularnewline
	\hline 
	Weighted intervals  & $1+\epsilon$  & $O_{\epsilon}(1)\lg^{2}n\lg^{5}N\lg W$ \tabularnewline
	\hline 
	Unweighted $d$-dimensional hypercubes  & $(1+\epsilon)2^{d}$  & $O_{d,\epsilon}(1)\log^{2d+1}n\log^{2d+1}N$\tabularnewline
	\cline{2-3} 
	 & $1+\epsilon$  & $n^{(1/\varepsilon)^{\Omega(1)}}$\tabularnewline
	\hline 
	Weighted $d$-dimensional hypercubes  & $(4+\epsilon)2^{d}$  & $O_{d,\epsilon}(1)\log^{2d-1}n\log^{2d+1}N\log W$\tabularnewline
	\hline 
	Weighted $d$-dimensional hyperrectangles  & $(1+\epsilon)\log^{d-1}N$  & $O_{d,\epsilon}(1)\lg^{2}n\log^{5}N\lg W$ \tabularnewline
	\hline 
	\end{tabular}
	\par\end{centering}
	\caption{
		Summary of our dynamic approximation algorithms and lower bounds. All
		algorithms are deterministic and work in the fully dynamic setting where
		in each update one interval/hypercube is inserted or deleted. We assume
		that all input objects are contained in $[0,N]^{d}$ and have weights in
		$[1,W]$; we do \emph{not} assume that the input objects have
		integer-coordinates.  Here, we write $O_{\varepsilon}(1)$ and
		$O_{d,\varepsilon}(1)$ to hide terms which only depend on $d$ and
		$\varepsilon$.}
	\label{table:results}
\end{table}

\textbf{Techniques.} Our main obstacle is that the maximum IS is a \emph{global}
property, i.e., when the input changes slightly, e.g., a single interval is
inserted or deleted, then it can cause a change of the optimal IS which
propagates through the entire instance (see Figure~\ref{fig:toy-example-1}).
Even worse, there are instances in which \emph{any }solution \emph{with a
non-trivial approximation guarantee} requires $\Omega(n)$ changes after an
update (see Figure~\ref{fig:toy-example-3}).

\begin{figure*}[t!]
  \includegraphics[width=14cm]{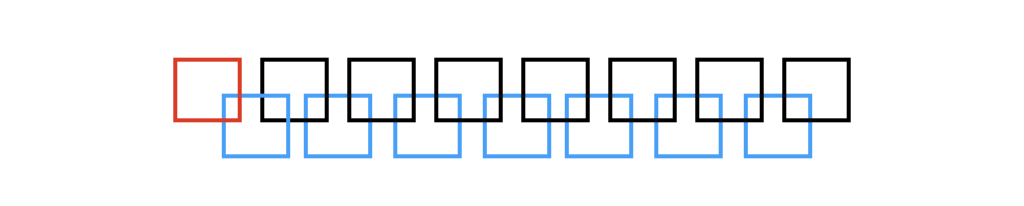}
  \caption{An instance of unweighted IS of intervals, i.e., all intervals have
		weight $1$. Observe that before inserting the red interval at the left,
		the solution consisting of the blue intervals in the bottom row is optimal. 
		However, after inserting the red interval, the optimal solution is
		unique and consists of the red interval together with all black
		intervals in the top row.}
  \label{fig:toy-example-1}
\end{figure*}

\begin{figure*}[t!]
  \includegraphics[width=14cm]{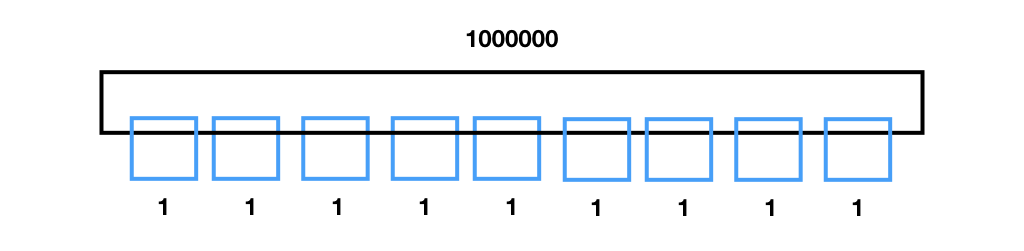}
  \caption{An instance of weighted IS of intervals. Note that when the
		large black interval with weight $1000000$ is present, any solution with
		non-trivial approximation ratio must contain the black interval.
		However, when the black interval is not present, the solution must
		contain many small blue intervals. Thus, inserting or deleting the black
		interval requires $\Omega(n)$ changes to the solution.}
  \label{fig:toy-example-3}
\end{figure*}

To limit the propagation effects, our algorithms for intervals
and hypercubes use a hierarchical grid decomposition. We partition
the space $[0,N]^{d}$ recursively into equally-sized grid cells with
$\log N$ levels, halving the edge length in each dimension of each cell when going
from one level to the next (similar to the quad-tree decomposition
in~\cite{arora1996polynomial}). Thus, each grid cell $Q$ has $2^{d}$
\emph{children }cells which are the cells of the respective next level
that are contained in $Q$. Also, each input object $C$ (i.e., interval
or hypercube) is contained in at most $\log N$ grid cells and it is \emph{assigned} to the grid level $\ell(C)$ in which the size
of the cells is ``comparable'' to the size of $C$. When an object $C$ is
inserted or deleted, we recompute the solution for each of the $\log N$
grid cells containing $C$, in a bottom-up manner. More precisely,
for each such cell $Q$ we decide which of the hypercubes assigned
to it we add to our solution, based on the solutions of the
children of $Q$. Thus, a change of the input does \emph{not
}propagate through our entire solution but only affects $\log N$
grid cells and the hypercubes assigned to them.

Also, we do not store the computed solution explicitly as this might require
$\Omega(n)$ changes after each update.
Instead, we store it \emph{implicitly.} In particular, in each
grid cell $Q$ we store a solution only consisting of objects assigned to $Q$ and
pointers to the solutions of children cells. Finally, at query time we output only
those objects that are contained in a solution of a cell $Q$ and
which do not overlap with an object in the solution of a cell of higher
level. In this way, if a long interval with large weight appears or
disappears, only the cell corresponding to the interval needs to be
updated, the other changes are done implicitly.

Another challenge is to design an algorithm that, given a cell
$Q$ and the solutions for the children cells of $Q$, computes an
approximate solution for $Q$ in time $\mathrm{poly}(\log n,\log N)$.
In such a small running time, we cannot afford to iterate over
all input objects assigned to $Q$. We now explain in more detail how our
algorithm overcome this obstacle.

\emph{Weighted hypercubes.} Let us first consider our $(4+\epsilon)2^{d}$-approximation
algorithm for weighted hypercubes. Intuitively, we consider the hypercubes
ordered non-decreasingly by size and add a hypercube $C$ to the IS
if the weight of $C$ is at least twice the total weight of all hypercubes
in the current IS overlapping with $C$. We then remove all hypercubes
in the solution that overlap with $C$.

To implement this algorithm in polylogarithmic time, we need to make
multiple adjustments. First, for each cell $Q$ we maintain a range
counting data structure $P(Q)$ which contains the (weighted) vertices
of all hypercubes that were previously selected in the IS solutions of children cells $Q'\subseteq Q$.
We will use $P(Q)$ to estimate the weight of hypercubes that
a considered hypercube $C$ overlaps with. Second, we use $P(Q)$
to construct an {\em auxiliary grid} within $Q$. The auxiliary
grid is defined such that in each dimension the grid contains $O_{d,\varepsilon}(\polylg N)$
grid slices; thus, there are $(\log N)^{O_{d,\varepsilon}(1)}$ subcells of $Q$
induced by the auxiliary grid. Third, we cannot afford to iterate over all hypercubes
contained in $Q$ to find the smallest hypercube $C$ that has at
least twice the weight $w'$ of the hypercubes in the current solution
that overlap with $C$. Instead, we iterate over all subcells $S\subseteq Q$
which are induced by the auxiliary grid and look for a hypercube
$C$ of large weight within $S$; we show that the total weight of
the points in $S\cap P(Q)$ is a sufficiently good approximation of $w'$.
If we find a hypercube $C$ with these properties, we add $C$ to
the current solution for $Q$, add the vertices of $C$ to $P(Q)$
and adjust the auxiliary grid accordingly. In this way, we need to
check only $(\log N)^{O_{d,\epsilon}(1)}$ subcells of $Q$ which we
can do in polylogarithmic time, rather than iterating over all hypercubes
assigned to $Q$. We ensure that for each cell $Q$ we need to repeat
this process only a polylogarithmic number of times. To show the approximation
bound we use a novel charging argument based on the points in $P(Q)$.
We show that the total weight of the points stored in $P(Q)$ estimates
the weight of the optimal solution for $Q$ up to a constant factor.
We use this to show that our computed solution is a $(4+\epsilon)2^{d}$-approximation.

\emph{Weighted intervals.} Next, we sketch our dynamic $(1+\varepsilon)$-approximation algorithm
for weighted IS of intervals. A greedy approach would be to build the solution
such that the intervals are considered in increasing order of their
lengths and then for each interval to decide whether we want to select
it and whether we want to remove some previously selected intervals
to make space for it. However, this cannot yield a $(1+\epsilon)$-approximate
solution. There are examples in which one can choose only one out
of multiple overlapping short intervals and the wrong choice implies that
one cannot obtain a $(1+\varepsilon)$-approximation together with
the long intervals that are considered later (see Figure~\ref{fig:toy-example}). However,
in these examples the optimal solution (say for a cell $Q$) consists
of only $O_{\epsilon}(1)$ intervals. Therefore, we show that in this
case we can compute a $(1+\epsilon)$-approximate solution in time
$O_{\epsilon}(\log^{2}n)$ by guessing the rounded weights of the intervals
in the optimal solution, guessing the order
of the intervals with these weights, and then selecting the intervals
greedily according to this order. On the other hand, if the optimal
solution for a cell $Q$ contains $\Omega_{\epsilon}(1)$ many intervals
with similar weights then we can take the union of the previously
computed solutions for the two children cells of $Q$. This sacrifies
at most one interval in the optimal solution for $Q$ that overlaps
with both children cells of $Q$ and we can charge this interval to
the $\Omega_{\epsilon}(1)$ intervals in the solutions for the children
cells of $Q$.

\begin{figure*}[t!]
  \includegraphics[width=14cm]{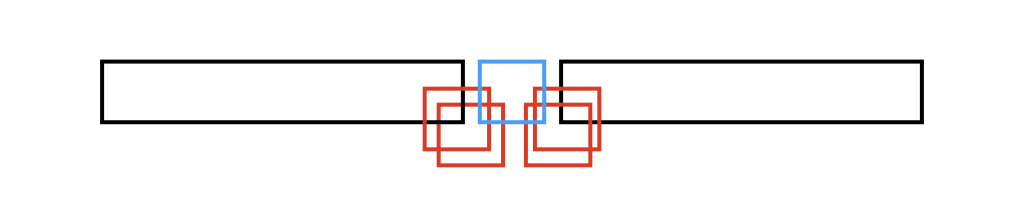}
  \caption{An instance of unweighted IS of intervals, i.e., all intervals have
		weight $1$. Observe that the optimal solution has size $3$ and consists of
		the small blue interval and the two large black intervals at the top. If
		an algorithm decides to pick any of the small red intervals, its
		solution can have size at most $2$.}
  \label{fig:toy-example}
\end{figure*}

Our algorithm interpolates between these two extreme cases. To this
end, we run the previously described $O(1)$-approximation algorithm
for hypercubes as a subroutine and use it to split each cell $Q$
into \emph{segments}, guided by the set $P(Q)$ above. Then we use that
for each set $S\subseteq Q$ the weight of $S\cap P(Q)$ approximates
the weight of the optimal solutions of intervals contained in $S$ within a
constant factor.
This is crucial for some of our charging arguments in which we show that 
the intervals contained in some sets $S$ can be ignored.
We show that for each cell $Q$ there is a $(1+\epsilon)$-approximate
solution in which $Q$ is partitioned into segments such that each
of them is either \emph{dense }or \emph{sparse. }Each dense segment
contains many intervals of the optimal solution and it is contained
in one of the children cells of $Q$. Therefore, we can
copy the previously computed solution for the respective child of
$Q$. Each sparse segment only contains $O_{\epsilon}(1)$ intervals
and hence we can compute its solution directly using guesses as described above.
In each level, this incurs an error and we use several involved charging
arguments to ensure that this error does not accumulate over the $\log N$
levels, but that instead it stays bounded by $1+\epsilon$.

\textbf{Other related work.}
Emek et al.~\cite{emek16space}, Cabello and
P{é}rez{-}Lantero~\cite{cabello17interval} and Bakshi et
al.~\cite{bakshi19weighted} study IS of intervals in the streaming model and
obtain algorithms with sublinear space usage.
In~\cite{emek16space,cabello17interval} insertion-only streams of unweighted
intervals are studied. They present algorithms which are
$(3/2+\varepsilon)$-approximate for unit length intervals and
$(2+\varepsilon)$-approximate for arbitrary-length intervals; they also provide
matching lower bounds.  Bakshi et al.~\cite{bakshi19weighted}~study turnstile
streams in which intervals can be inserted and deleted. They obtain algorithms
which are $(2+\varepsilon)$-approximate for weighted unit length intervals and
a $O(\lg n)$-approximate for unweighted arbitrary length intervals; they also
prove matching lower bounds.

In two dimensions, Agarwal et al.~\cite{agarwal98label} presented a static
algorithm which computes $O(\log n)$~approximation of the maximum IS of $n$
arbitrary axis-parallel rectangles in time $O(n\log n)$. They also show how to
compute a $(1+1/k)$-approximation of unit-height rectangles in time
$O(n\log n+n^{2k-1})$ for any integer $k\ge1$.

\textbf{Problem definition and notation.}
We assume that we obtain a set $\C=\{C_{1},\dots,C_{n}\}$ of $d$-dimensional
hyperrectangles in the space $[0,N]^{d}$ for some global value $N\in\R$.
Each hyperrectangle $C_{i}\in\C$ is characterized by coordinates
$x_{i}^{(1)},y_{i}^{(1)},\dots,x_{i}^{(d)},y_{i}^{(d)}\in[0,N]$ such
that $C_{i}:=(x_{i}^{(1)},y_{i}^{(1)})\times\dots\times(x_{i}^{(d)},y_{i}^{(d)})$
and a weight $w_{i}\in[1,W]$ for some global value $W$; we do \emph{not} assume that
the coordinates of the input objects are integer-valued.
We assume
that $1\leq y_{i}^{(j)}-x_{i}^{(j)}\leq N$ for each $j\in[d]$. If $C_{i}$
is a hypercube then we define $s_{i}$ such that $s_{i}=y_{i}^{(j)}-x_{i}^{(j)}$
for each dimension $j$. Two hypercubes $C_{i},C_{i'}\in\C$ with
$i\ne i'$ are \emph{independent} if $C_{i}\cap C_{i'}=\emptyset$.
Note that we defined the hypercubes as open sets and, hence, two dependent hypercubes
cannot overlap in only a single point.
A set of hyperrectangles $\C'\subseteq\C$ is an \emph{independent
set (IS)} if each pair of hypercubes in $\C'$ is independent. The \emph{maximum
IS} problem is to find an IS $\C'\subseteq \C$, that maximizes
$w(\C'):=\sum_{C_{i}\in\C'}w_{i}$.

Due to space constraints we present some of the results in the appendix; missing
proofs can be found in Appendix~\ref{sec:omitted}.
Appendix~\ref{sec:overview-appendix} gives an overview over the contents of
the appendix.

\section{Hierarchical Grid Decomposition}
\label{sec:framework}

We describe a hierarchical grid decomposition that we use for all
our algorithms for hypercubes (for any $d$), that is similar to \cite{arora1996polynomial}. It is based on a hierarchical grid $\G$ over the space $[0,N]^{d}$
where we assume w.l.o.g.~that $N$ is a power of 2 and $N$ is an
upper bound on the coordinates of every object in each dimension.
The grid $\G$ has $\log N$ \emph{levels. }In each level, the space
$[0,N]^{d}$ is divided into \emph{cells}; the union of the cells
from each level spans the whole space. There is one grid cell of level
0 that equals to the whole space $[0,N]^{d}$. Essentially, each grid
cell of a level $\ell<\log N$ contains $2^{d}$ grid cells of level
$\ell+1$. We \emph{assign} the input hypercubes to the grid cells.
In particular, for a grid cell $Q\in\G$ we assign a set $\C'(Q)\subseteq\C$ to
$Q$ which are all input hypercubes that are contained in $Q$ and
whose side length is a $\Theta(\epsilon/d)$-fraction of the side
length of $Q$ (we will make this formal later). This ensures the helpful
property that any
IS consisting only of hypercubes in $\C'(Q)$ has size at most $O\left((\frac{d}{\epsilon})^{d}\right)$.
For each cell $Q$ we define $\C(Q):=\bigcup_{Q':Q'\subseteq Q}\C'(Q)$
which are all hypercubes contained in $Q$. One subtlety
is that there can be input hypercubes that are not assigned to any
grid cell, e.g., hypercubes that are very small but overlap more than
one very large grid cell. Therefore, we shift the grid by some offset
$a\in[0,N]$ in each dimension which ensures that those hypercubes
are negligible.

Formally, let $\epsilon>0$ such that $1/\epsilon$ is an integer
and a power of 2. For each $\ell\in\{0,\dots,\log N\}$ let $\G_{\ell}$
denote the set of grid cells of level $\ell$ defined as $Q_{\ell,k}:=[0,N]^{d}\cap\prod_{j=1}^{d}[a+k^{(j)}\cdot N/2^{\ell-1},a+(k^{(j)}+1)\cdot N/2^{\ell-1}]$
for each $k=(k^{(1)},\dots,k^{(d)})\in\Z^{d}$. Then $\G_{0}$ consists
of only one cell $[0,N]^{d}=:Q^{*}$.We define $\G:=\bigcup_{\ell=0}^{\log N}\G_{\ell}$.
For a grid cell $Q\in\G$, we let $\ell(Q)$ denote the level of $Q$
in $\G$. Note that for each cell $Q$ of level $\ell(Q)<\lg N$,
there are at most $2^{d}$ grid cells $Q_{i}$ of level $\ell(Q_{i})=\ell(Q)+1$
and that are contained in $Q$, i.e., such that $Q_{i}\subseteq Q$.
We call the latter cells the \emph{children} of $Q$ and denote them
by $\ch(Q)$. Informally, a hypercube $C_{i}$ has \emph{level $\ell$} if 
$s_{i}$ is within a $\Theta(\epsilon/d)$-fraction of the side length of
grid cells of level $\ell$; formally, $C_{i}$ has level \emph{$\ell$}
if $s_{i}\in[\epsilon N/(d2^{\ell-1}),2\cdot\epsilon N/(d2^{\ell-1}))$
for $\ell=1,\dots,\lg N$ and $s_{i}\in[\epsilon N/d,N]$ for $\ell=0$.
For each $C\in\C$ denote by $\ell(C)$ the level of $C$. We \emph{assign}
a hypercube $C$ to a cell $Q$ if $C_{i}\subseteq Q$ and $\ell(C)=\ell(Q)$;
the set of all these hypercubes for a cell $C$ is defined by $\C'(Q):=\{C_{i}\in\C|C_{i}\subseteq Q\wedge\ell(C)=\ell(Q)\}$.
For each grid cell $Q$ we define $\C(Q)$ to be the set of all hypercubes
contained in $Q$ that are assigned to $Q$ or to grid cells contained
in $Q$, i.e., $\C(Q):=\bigcup_{Q':Q'\subseteq Q}\C'(Q)$.

For each cell $Q$, we partition the hypercubes in $\C(Q)$ and $\C'(Q)$
based on their weights in powers of $1+\varepsilon$. For each $k\in\Z$
we define $\C_{k}:=\{C_{i}\in\C:w_{i}\in[(1+\epsilon)^{k},(1+\epsilon)^{k+1})\}$
and for each grid cell $Q$ we define $\C_{k}(Q):=\C_{k}\cap\C(Q)$
and $\C'_{k}(Q):=\C{}_{k}\cap\C'(Q)$. Note that $\C_{k}=\emptyset$
if $k<0$ or $k>\log_{1+\epsilon}W$.

In the next lemma we prove that there is a value for the offset $a$
such that there is a $(1+\epsilon)$-approximate solution $\OPT'$
that is \emph{grid-aligned, }i.e., for each $C\in\OPT'$ there is
a grid cell $Q$ in the resulting grid for $a$ such that $C\in\C'(Q)$.

\global\long\def\off{\mathrm{off}}%

\begin{lem}
\label{lem:offset} In time $(d/\epsilon)^{O(1/\epsilon)}\log N$
we can compute a set $\off(\epsilon)$ with $\left|\off(\epsilon)\right|\le(d/\epsilon)^{O(1/\epsilon)}$
that is independent of the input objects $\C$ and that contains an
offset $a\in\off(\epsilon)$ for the grid for which the optimal grid-aligned
solution $\OPT'$ satisfies that $w(\OPT')\ge(1-O(\epsilon))w(\OPT)$.
If we draw the offset $a$ uniformly at random from $\off(\epsilon)$, then
$\mathbb{E}[w(\OPT')]\ge(1-O(\epsilon))w(\OPT)$ and
$w(\OPT')\ge(1-O(\epsilon))w(\OPT)$ with constant
probability.
\end{lem}

For the deterministic results in this paper we run our algorithms
for each choice of $a\in\off(\epsilon)$ in parallel and at the end
we output the solution with maximum weight over all choices of $a$.
For our randomized results we choose $O(\log N)$ offsets $a\in\off(\epsilon)$
uniformly at random and hence there exists a grid-aligned solution
$\OPT'$ with $w(\OPT')\ge(1-O(\epsilon))w(\OPT)$ with high probability
(i.e., with probability at least $1-\left(1/N\right)^{O(1)}$).
\begin{lem}
\label{lem:properties-grid} Each grid cell $Q\in\G$ has a volume
of $(N/2^{\ell(Q)-1})^{d}$ and can contain at most $(d/\epsilon)^{d}$
independent hypercubes from $\C'(Q)$. Also, each hypercube $C_{i}\in\C$
is contained in $\C'(Q')$ for at most one grid cell $Q'$ and in
$\C(Q')$ for at most $\log N$ cells $Q'$. 
\end{lem}

\textbf{Data structures.} \label{sec:data-structure-cell} We define
a data structure which will allow us to access the hypercubes in the
sets $\C(Q)$, $\C'(Q)$, etc. for each cell $Q$ efficiently. Roughly
speaking, these data structures let us insert and delete hypercubes
and answer queries of the type: ``Given a hyperrectangle~$B$, return
a hypercube which is contained in $B$.'' They are constructed using
data structures for range counting/reporting~\cite{lee84computational,willard85adding,edelsbrunner1981note}.
\begin{lem}
\label{lem:data-structure} Let $d\in\N$. There is a data structure
that maintains a set $\C'$ of weighted hypercubes in $\R^{d}$ and
allows the following operations: 
\begin{enumerate}
\item Initialize the data structure, report whether $\C'=\emptyset$, both
in worst-case time $O(1)$. 
\item Insert or delete a hypercube into (from) $\C'$ in worst-case time
$O(\log^{2d+1}|\C'|)$. 
\item \label{item:cuboid-query} For a hyperrectangle $B\subseteq\R^{d}$,
check whether there is a hypercube $C_{i}\in\C'$ with $C_{i}\subseteq B$
in time $O(\log^{2d-1}|\C'|)$. If yes, return one such hypercube
in time $O(\log^{2d-1}|\C'|)$ and the smallest such hypercube (i.e.,
with smallest size $s_{i}$) in time $O(\log^{2d+1}|\C'|)$. 
\item If $d=1$, given a value $t\in\mathbb{R}$, return the element $C_{i}=(x_{i}^{(1)},y_{i}^{(1)})$
with minimum value $y_{i}^{(1)}$ among all elements $C_{i'}=(x_{i'}^{(1)},y_{i'}^{(1)})$
with $t\le x_{i'}^{(1)}$, in time $O(\lg^{2}n)$. 
\end{enumerate}
\end{lem}

Using Lemma~\ref{lem:data-structure} for each cell $Q$ we define
data structures $\D(Q)$, $\D'(Q)$, $\D_{k}(Q)$, and $\D'_{k}(Q)$
for maintaining the sets $\C(Q)$, $\C'(Q)$, $\C_{k}(Q)$, and $\C_{k}'(Q)$
for each $k=1,\dots,\lg_{1+\varepsilon}W$, respectively, where $W$
is an upper bound on the maximum weight of all hypercubes. The grid
$\G$ as defined above contains $\Omega(N^{d})$ cells in total. However,
there are only $O(n\lg N)$ cells in $\G$ such that $\C(Q)\ne\emptyset$
(by Lemma~\ref{lem:properties-grid}), denote them by $\G'$. We
use a data structure that maintains these cells $\G'$ such that in
worst-case time $O(\lg|\G'|)$ we can add and remove a cell, get pointers
to the data structures $\D(Q)$, $\D'(Q)$, $\D_{k}(Q)$, $\D'_{k}(Q)$
for a cell $Q$, and get and set pointers to a solution that we compute
for a cell $Q$. See Appendix~\ref{apx:data-struct-maintain-grid}
for details.

\textbf{Algorithmic framework.} Now we sketch the framework for implementing our
dynamic algorithms. Due to space constraints we postpone its formal definition
to Appendix~\ref{sec:framework-details:unweighted}.

For each cell $Q$ we maintain a solution $\ALG(Q)\subseteq\C(Q)$ that is
near-optimal, i.e., with $w(\OPT(Q))\le\alpha\cdot w(\ALG(Q))$
for the approximation ratio $\alpha$ of the respective setting. We ensure that
$\ALG(Q)$ depends only on $\C(Q)$ and not on hypercubes $C$ with
$C\notin\C(Q)$.

To implement update operations, we proceed as follows. When a hypercube $C$ is inserted or deleted,
we need to update only the solutions $\ALG(Q)$ for the at most $\log N$
cells $Q$ such that $C\in\C(Q)$. We will update the solutions $\ALG(Q)$
in a bottom-up manner, i.e., we order the cells $Q$ with $C\in\C(Q)$
decreasingly by level and update their respective solutions $\ALG(Q)$
in this order. To ensure a total update time of $(\log n+\log N)^{O_{d,\varepsilon}(1)}$,
we will define algorithms that update $\ALG(Q)$ for a cell $Q$ in
time $(\log n+\log N)^{O_{d,\varepsilon}(1)}$, given that we already updated the
solutions $\ALG(Q')$ for all cells $Q'\varsubsetneq Q$. In fact,
we will essentially re-compute the solution $\ALG(Q)$ for a cell $Q$
from scratch, using only the solutions $\left\{ \ALG(Q')\right\} _{Q'\in\ch(Q)}$
computed for the children of $Q$.

Finally, to implement query operations, i.e., to output an approximate solution for
the whole space $[0,N]^d$, we return the solution $\ALG(Q^{*})$ (recall that
		$Q^*$ is the grid cell at level $0$ which contains the whole space).
We will show in the respective sections how we can output the weight of
$\ALG(Q^*)$ in time $O_{d,\varepsilon}(1) \poly(\lg n, \lg N)$ $\ALG(Q^*)$ and
how to output all hypercubes in $\ALG(Q^*)$ in time
$O_{d,\varepsilon}(|\ALG(Q^*)| \poly(\lg n, \lg N))$.

\section{Weighted Hypercubes}
\label{sec:Weighted-hypercubes}

We study now the weighted case for which we present a dynamic
$(4+\epsilon)2^{d}$-approximation algorithm for $d$-dimensional hypercubes.  Our
strategy is to mimic a greedy algorithm that sorts the hypercubes by size
$s_{i}$ and adds a hypercube $C_{i}$ with weight $w_i$ if it does not overlap
with any previously selected hypercube or if the total weight of the previously
selected hypercube that $C_{i}$ overlaps with is at most $w_{i}/2$. Using a
charging argument one can show that this yields a $2^{d+2}$-approximate
solution. The challenge is to implement this approach such that we obtain
polylogarithmic update time.

From a high-level point of view, our algorithm works as follows. 
In each cell $Q$, we maintain a set of points $P(Q)$ containing the vertices of
all hypercubes which have been added to independent sets $\ALG(Q')$ for cells
$Q'\subset Q$.  The weight of each point is the weight of the corresponding
hypercube.  Based on the points in $P(Q)$, we construct an auxiliary grid inside
$Q$ which allows to perform the following operation efficiently: ``Given a set
of auxiliary grid cells $A$, find a hypercube $C\in\C'(Q)$ in $A$ whose weight is
at least twice the weight of all points in $P(Q) \cap A$.'' 
When we try to add a hypercube to $Q$ we do not iterate over all hypercubes
contained in $Q$ but instead enumerate a polylogarithmic number of sets $A$ and 
perform the mentioned query for each of them. Also, we 
do not maintain the current independent set explicitly
(which might change a lot after an update), but we 
update only the weight of the points in $P(Q) \cap A$, which can be done efficiently.
For each cell $Q$ we add
only a polylogarithmic number of hypercubes to $\ALG(Q)$. If a hypercube $C_i \in \ALG(Q)$
overlaps with a hypercube $C_{i'} \in \ALG(Q')$ for some cell $Q'\subset Q$ then we exclude $C_{i'}$
from the solution that we output, but do not delete $C_{i'}$ from $\ALG(Q')$. In this way, we obtain polylogarithmic
update time, even if our computed solution changes a lot.

Before we describe our algorithm in detail, let us first elaborate on how we
maintain the points $P(Q)$.  In the unweighted settings, for each cell $Q$ we
stored in $\ALG(Q)$ a set of hypercubes or pointers to such sets. Now, we define
each set $\ALG(Q)$ to be a pair $(\bar{\C}(Q),P(Q))$. Here,
$\bar{\C}(Q)\subseteq\C'(Q)$ is a set of hypercubes from $\C'(Q)$ that
we selected for the independent set (recall that $\C'(Q)$ contains the
hypercubes $C \subseteq Q$ with $\ell(C_i) = \ell(Q)$); and $P(Q)$ is the
data structure for the range counting/reporting problem according to
Lemma~\ref{lem:range-counting-data-structure}. We will often
identify $P(Q)$ with the set of points stored in $P(Q)$. 
\begin{lem}[\cite{lee84computational,willard85adding,edelsbrunner1981note}]
\label{lem:range-counting-data-structure}
	There exists a data structure
	that maintains a set of weighted points $P\subseteq\R^{d}$ and allows
	the following operations: 
	\begin{itemize}
		\item add or delete a point in $P$ in worst-case time $O(\log^{d}|P|)$, 
		\item report or change the weight of a point in $P$ in worst-case time
			$O(\log^{d}|P|)$, 
		\item given an open or closed hyperrectangle $B\subseteq\R^{d}$, report
			the total weight of the points $B\cap P$, in worst-case time $O(\log^{d-1}|P|)$. 
		\item given $d'\in[d]$ and an interval $I=[x,z]\subseteq\R$, in worst-case
			time $O(\log|P|)$ report a value $y$ such that at most
			$\Gamma:=\left|P\cap\left(\R^{d'-1}\times[x,z]\times\R^{d-d'}\right)\right|/2$
			points are contained in $\R^{d'-1}\times(x,y)\times\R^{d-d'}$ and at
			most $\Gamma$ points are contained in
			$\R^{d'-1}\times(y,z)\times\R^{d-d'}$. 
	\end{itemize}
\end{lem}

Now we describe our algorithm in detail.
Let again $x_{Q}^{(1)},\dots,x_{Q}^{(d)}$ and $y_{Q}^{(1)},\dots,y_{Q}^{(d)}$
be such that $Q=[x_{Q}^{(1)},y_{Q}^{(1)}]\times\dots\times[x_{Q}^{(d)},y_{Q}^{(d)}]$.
We construct a data structure $P(Q)$ according to Lemma~\ref{lem:range-counting-data-structure}
such that initially it contains the points $\bigcup_{Q'\in\ch(Q)}P(Q')$;
this will ensure that initially the points in $P(Q)$ are the vertices of all hypercubes in
$\bar{\C}(Q')$ for each $Q' \subset Q$.
Constructing $P(Q)$ might take more than polylogarithmic time since
the sets $P(Q')$ with $Q'\in\ch(Q)$ might contain more than polylogarithmically
many points. However, we show in Appendix~\ref{sec:adjustment} how to adjust our
hierarchical grid decomposition and the algorithm to obtain polylogarithmic update time.

We want to compute a set $\bar{\C}(Q)\subseteq\C'(Q)$ containing the
hypercubes from $\C'(Q)$ that we add to the independent set.
At the beginning, we initialize $\bar{\C}(Q):=\emptyset$.
We compute an auxiliary grid $(Z^{(1)},\dots,Z^{(d)})$ in order to search
for hypercubes to insert, similar to the unweighted case. To
define this auxiliary grid, we first compute the total weight $\tilde{W}=w(P(Q))$
of all points that are in $P(Q)$ at the beginning of the algorithm in time
$O(\log^{d-1} |P(Q)|)$, where we define $w(P):=\sum_{p\in P}w_{p}$ for any set of weighted points $P$.
Then we define the auxiliary grid within $Q$ such that \emph{in the interior}
of each \emph{grid slice} the points in $P(Q)$ have a
total weight at most $\epsilon^{d+2}\tilde{W}/(d^{d+1} \log N)$, where a grid slice is
a set of the form $\R^{d'-1}\times(x,y)\times\R^{d-d'}$ for some $d'\in[d]$. We emphasize here that
this property only holds for the \emph{interior} of the grid slices and that the
sets in the first point of Lemma~\ref{lem:construct-aux-grid} are open.
\begin{lem}
\label{lem:construct-aux-grid}
	Given a cell $Q$ and the data structure $P(Q)$, in 
	$O\left(\left(\frac{d}{\epsilon}\right)^{d+2} \cdot\log^{d}|P(Q)|\cdot\log N\right)$
	time we can compute sets $Z^{(1)},\dots,Z^{(d)}$ of coordinates with
	$Z^{(d')}=\{z_{1}^{(d')},z_{2}^{(d')},\dots\}$
	for each $d'$ such that 
	\begin{itemize}
		\item the total weight of the points in
			$\left(\R^{d'-1}\times(z_{j}^{(d')},z_{j+1}^{(d')})\times\R^{d-d'}\right)\cap P(Q)$
			is at most $\epsilon^{d+2}\tilde{W}/(d^{d+1} \log N)$ for each $d'\in[d]$
			and each $j\in\{1,\dots,|Z^{(d')}|-1\}$,
		\item $Q=\prod_{d'=1}^{d}[z_{1}^{(d')},z_{|Z^{(d')}|}^{(d')}]$, and 
		\item $|Z^{(d')}|\le d^{d+1} \log N/\epsilon^{d+2}+1$ for each $d'\in[d]$. 
	\end{itemize}
\end{lem}

To select hypercubes to add to $\bar{\C}(Q)$ our algorithm runs in iterations,
and in each iteration we add one hypercube to $\bar{\C}(Q)$. In each
iteration we enumerate all hyperrectangles $A\subseteq Q$ that are aligned
with $Z^{(1)},\dots,Z^{(d)}$; note that there are only $\left(d^{d+1} \log N/\epsilon^{d+2}+1\right)^{2d}$
such hyperrectangles (by the third point of Lemma~\ref{lem:construct-aux-grid}). For each
such hyperrectangle $A$ we use the data structures
$\left\{ \D'_{k}(Q) \right\} _{k\in\N}$ to determine whether there is
a hypercube $C_{i}\subseteq A$ contained in $\C_{k}'(Q)$ for some
$k$ such that $(1+\epsilon)^{k}\ge2w(P(Q)\cap A)$ (recall that
$\D'_{k}(Q)$ maintains the intervals in $\C'$ which have weights in the range
$[(1+\varepsilon)^k, (1+\varepsilon)^{k+1})$ and also recall that $\D'_{k}(Q)$
is contained in the input $\DD(Q)$ of the algorithm as discussed in
Section~\ref{sec:framework}).
We say that such
a hypercube $C_{i}$ is \emph{addible}. If there is no addible hypercube
$C_{i}$ then we stop and return $(\bar{\C}(Q),P(Q))$. Otherwise,
we determine the smallest addible hypercube $C_{i}$ (i.e., with minimum
value $s_{i}$) and we add $C_{i}$ to $\bar{\C}(Q)$.
We add to $P(Q)$ the $2^{d}$ vertices of $C_{i}$ with weight $w_{i}$;
if a vertex of $C_{i}$ has been in $P(Q)$ before then we increase its
weight by $w_{i}$. We remove from $\bar{\C}(Q)$ all hypercubes that
$C_{i}$ overlaps with. Intuitively, we remove also all other previously
selected hypercubes that $C_{i}$ intersects; however, we do not do
this explicitly since this might require $\Omega(n)$ time, but we
will ensure this implicitly via the query algorithm that we use to output the
solution and that we define below. Finally, we add the coordinates of $C_{i}$
to the coordinates of the grid $Z^{(1)},\dots,Z^{(d)}$, i.e., we make
the grid finer; formally, for each $d'\in[d]$ we add to $Z^{(d')}$
the coordinates $\{x_{i}^{(d')},y_{i}^{(d')}\}$. This completes one
iteration. 
\begin{lem}
\label{lem:iterations}
	The algorithm runs for at most $\left(\frac{d}{\epsilon}\right)^{d} \log W$
	iterations and computes $\bar{\C}(Q)$ in time
	$O\left(\left(\frac{d}{\epsilon}\right)^{2d^{2}+5d+1}\cdot\lg W\cdot\log^{2d-1}n\log^{2d}N\right)$. 
\end{lem}

After the computation above, we define that our solution $\SOL(Q)$ for $Q$ contains
all hypercubes in a set $\bar{\C}(Q')$ for some cell $Q'\subseteq Q$ that are not
overlapped by a hypercube in a set $\bar{\C}(Q'')$ for some cell $Q''\supset Q'$.
So if two hypercubes $C_{i'}\in\bar{\C}(Q')$,
$C_{i''}\in\bar{\C}(Q'')$ overlap and $\ell(Q')<\ell(Q'')$,
then we select $C_{i'}$ but not $C_{i''}$. We can output $\SOL(Q)$ in time $O_{d,\epsilon}(|\SOL(Q)|\log^{d+1}N)$, 
see Lemma~\ref{lem:output-routine-weighted-hypercubes} in the appendix.
If we only want return the approximate weight of $\SOL(Q)$, we can return
$w(P(Q))$ which is a $O(2^d)$-approximation by
Lemma~\ref{lem:apx-ratio-weighted-hypercubes} below. 

Finally, we bound our approximation ratio. Whenever we add a hypercube
$C_{i}$ to a set $\bar{\C}(Q)$ for some cell $Q$, then we explicitly
or implicitly remove from our solution all hypercubes $C_{i'}$ with
$C_{i}\cap C_{i'}\ne\emptyset$ such that $C_{i'}\in\bar{\C}(Q')$
for a cell $Q'\subseteq Q$ . However, the total weight of these removed
hypercubes is bounded by $w_{i}/2$ since in the iteration in which
we selected a hypercube $C_{i}\in\C_{k}$ there was a set $A\supseteq C_{i}$
with $(1+\epsilon)^{k}\ge2w(p(Q)\cap A)$ and by definition of $\C_k$, 
$w_i \geq (1+\varepsilon)^k$. Therefore, we can bound
our approximation ratio using a charging argument. 

\begin{lem}
\label{lem:apx-ratio-weighted-hypercubes}
	For each cell $Q\in\G'$,
	we have that
	\begin{align*}
		w(\SOL(Q))\le w(\OPT(Q))\le(2+O(\epsilon))w(P(Q))\le(4+O(\epsilon))2^{d}w(\SOL(Q))). 
	\end{align*}
\end{lem}
Before we prove Lemma~\ref{lem:apx-ratio-weighted-hypercubes}, we prove
three intermediate results. To this end, 
recall that $\SOL(Q)$ consists of hypercubes in sets $\bar{\C}(Q')$
for cells $Q'\subseteq Q$. First, we show via a token argument that the
total weight of \emph{all} hypercubes of the latter type is at most
$2w(\SOL(Q))$, using that when we inserted a new hypercube in our solution
then it overlapped with previously selected hypercubes of weight at
most $w_{i}/2$.
\begin{lem}
\label{lem:point-sol}
	We have that $w(P(Q)) \leq 2^{d+1} w(\SOL(Q))$.
\end{lem}
\begin{proof}
	We assign to each hypercube $C_{i}\in\SOL(Q)$ a budget
	of $2w_{i}$. We define now an operation that moves these budgets.
	Assume that a hypercube $C_{i'}\in\bar{\C}(Q')$ for some cell $Q'$
	now has a budget of $2w_{i}$ units. For each hypercube $C_{i''}\in\bar{\C}(Q'')$
	for some cell $Q''\subsetneq Q'$ such that one of the vertices of
	$C_{i''}$ is overlapped by $C_{i'}$, we move $2w_{i''}$ units of
	the budget of $C_{i'}$ to the budget of $C_{i''}$. Note that a hypercube
	$C_{i''}\in\C'(Q'')$ in a cell $Q''\subsetneq Q'$ overlaps $C_{i'}$
	if and only if $C_{i'}$ overlaps a vertex of $C_{i''}$ since $s_{i''}<s_{i'}$.
	When we selected $C_{i'}\in\C_{k}(Q')$ then there was a corresponding
	set $A\subseteq Q'$ such that $w_{i}\ge(1+\epsilon)^{k}\ge2w(p(Q')\cap A)$.
	Therefore, when we move the budget of $C_{i'}$ as defined then $C_{i'}$
	keeps $w_{i'}$ units of its budget. After this operation, we say
	that $C_{i'}$ is \emph{processed}. We continue with this operation
	until each hypercube $C_{i'}$ with a positive budget is processed.
	At the end, each hypercube $C_{i'}$ such that $C_{i'}\in\bar{\C}(Q')$
	for some cell $Q'$ has a budget of $w_{i}$. Therefore, 
	$\sum_{Q'\subseteq Q}\sum_{C_{i}\in\bar{\C}(Q')}w_{i}\le2w(\SOL(Q))$.

	Given the previous inequality and since we insert $2^d$ points for each
	$C_i\in\bar{C}(Q')$, we obtain that
	$w(P(Q)) =2^{d}\cdot2\sum_{Q'\subseteq Q}\sum_{C_{i}\in\bar{\C}(Q')}w_{i} \le2^{d+1}w(\SOL(Q))$.
\end{proof}
We want to argue that $w(\OPT(Q))\le(4+O(\epsilon))\cdot2^{d}w(\SOL(Q))$.
To this end, we classify the hypercubes in $\OPT(Q)$. For each $C_{i}\in\OPT(Q)$
such that $C_{i}\in\C'(Q')$ for some grid cell $Q'\subseteq Q$ we
say that $C_{i}$ is \emph{light} if $w_{i}\le w(P(Q'))\epsilon^{d+1}/(d^d \log N)$
and \emph{heavy} otherwise (for the set $P(Q')$ when the algorithm
finishes).

Next, we show that the total weight of light hypercubes is
$\varepsilon w(P(Q))$. We do this by observing that since each cell
$Q'\subseteq Q$ contains at most $(d/\epsilon)^{d}$ light hypercubes in
$\OPT(Q)\cap\C'(Q')$ (by Lemma~\ref{lem:properties-grid}), we can charge their
weights to $w(P(Q))$.
\begin{lem}
\label{lem:light}
	The total weight of light hypercubes is at most
	$\varepsilon w(P(Q))$.
\end{lem}
\begin{proof}
	Let $Q'\subseteq Q$. For each light hypercube
	$C_{i}\in\C'(Q')\cap\OPT(Q)$ we charge $w_{i} \leq w(P(Q'))\epsilon^{d+1}/(d^d \log N)$
	to the points $p\in P(Q')$, proportionally to their respective
	weight $w_p$. There are at most $(d/\epsilon)^{d}$ light hypercubes in
	$\C'(Q') \cap \OPT(Q)$ (by Lemma~\ref{lem:properties-grid}).
	Hence, the total charge is at most $w(P(Q')) \cdot \varepsilon / \lg N$ and
	each point $p\in P(Q')$ receives a total charge of at
	most $w_p \cdot \epsilon/\log N$ for $Q'$. Each point $p$ is contained in at most $\log N$
	sets in $\left\{ P(Q')\right\} _{Q'\subseteq Q}$. Therefore,
	the total weight of all light hypercubes in $\OPT(Q)$ is bounded by
	$\varepsilon w(P(Q))$.
\end{proof}

For the heavy hypercubes, we pretend that we increase the weight of
each point in $P(Q)$ by a factor $2+O(\epsilon)$. We show that then
each heavy hypercube $C_{i}\in\OPT(Q)$ contains points in $P(Q)$
whose total weight is at least $w_{i}$. Hence, after increasing the
weight, the weight of the points in $P(Q)$ ``pays'' for all heavy
hypercubes in $\OPT(Q)$.

Let $\beta:=1/(\frac{1}{2+2\epsilon}-2\epsilon)=2+O(\epsilon)$. For
each cell $Q'\subseteq Q$ and for each hypercube $C_{i}\in\bar{\C}(Q')$
we place a weight of $\beta w_{i}$ essentially on each vertex of
$C_{i}$. Since each hypercube $C_{i}$ is an open set, $C_{i}$
does not contain any of its vertices. Therefore, we place the weight
$\beta w_{i}$ not exactly on the vertices of $C_{i}$, but on the
vertices of $C_{i}$ slightly perturbed towards the center of $C_{i}$.
Then, the weight we placed for the vertices of $C_{i}$ contributes
towards ``paying'' for $C_{i}$. Formally, for a small value $\delta>0$
we place a weight of $\beta w_{i}$ on each point of the form $(x_{i}^{(1)}+k^{(1)}s_{i}+\delta-k^{(1)}\cdot2\delta,\dots,x_{i}^{(d)}+k^{(d)}s_{i}+\delta-k^{(d)}\cdot2\delta)$
with $k^{(d')}\in\{0,1\}$ for each $d'\in[d]$. We choose $\delta$
such that any input hypercube $C_{i'}$ overlaps each point $(x_{i}^{(1)}+k^{(1)}s_{i}+\delta-k^{(1)}\cdot2\delta+,\dots,x_{i}^{(d)}+k^{(d)}s_{i}+\delta-k^{(d)}\cdot2\delta)$
corresponding to $C_{i}$ if and only if $C_{i}\cap C_{i'}\neq\emptyset$.
We say that these points are the \emph{charge points }of $C_{i}$.
If on one of these points we already placed some weight then we increase
its weight by $\beta w_{i}$. Let $\tilde{P}(Q')$ denote the points
on which we placed a weight in the above procedure for $Q'$ or for
a cell $Q''\subseteq Q'$. For each point $p\in\tilde{P}(Q)$ let
$\tilde{w}_{p}$ denote the total weight that we placed on $p$ in
this procedure. Since for each point $p_{i}\in P(Q')$ with weight
$w_{i}$ we introduced a point $\tilde{p}_{i}\in\tilde{P}(Q')$ with
weight $\tilde{w}_{i}\geq\beta w_{i}\geq w_{i}$, we have that $\sum_{p\in P(Q')}w_{p}\le\sum_{p\in\tilde{P}(Q')}\tilde{w}_{p}$.%

\begin{lem}
\label{lem:heavy}
	The total weight of heavy hypercubes is at most
	$(2+O(\varepsilon)) w(P(Q))$.
\end{lem}
\begin{proof}
	Let $C_{i}\in\OPT(Q)$ be a heavy hypercube. We claim that for each heavy
	hypercube $C_{i}\in\OPT(Q)$ it holds that
	$\sum_{p\in \tilde{P}(Q) \cap C_{i}}\tilde{w}_{p}\ge w_{i}$.
	This implies the claim since
	$(2+O(\varepsilon)) \sum_{p\in P(Q) \cap C_i} w_p \geq \sum_{p\in \tilde{P}(Q) \cap C_i} \tilde{w}_p$.

   	Let $Q'$ denote the cell such that $C_{i}\in\C'(Q')$.
	Let $\tilde{\C}(Q')$ denote the hypercubes that are in the set $\bar{\C}(Q')$
	at some point while the algorithm processes the cell $Q'$. If $C_{i}\in\tilde{\C}(Q')$
	then the claim is true since we placed a weight of $\beta w_{i}$
	on essentially each of its vertices (slightly perturbed by $\delta$
	towards the center of $C_{i}$). Assume that $C_{i}\notin\tilde{\C}(Q')$.
	Let $k$ be such that $C_{i}\in\C_{k}(Q')$. 
	Consider the first iteration when we processed $Q'$ such that we added a
	hypercube $C_{i'}$ with size $s_{i'} > s_i$ or the final iteration if no
	hypercube with size larger $s_i$ is added. 
	Let $A\subseteq Q'$ denote the smallest
	set that is aligned with the auxiliary grid $Z^{(1)},\dots,Z^{(d)}$
	for the cell $Q'$ such that $C_{i}\subseteq A$.
	If $(1+\epsilon)^{k}\ge2w(P(Q')\cap A)$ then in this iteration we
	would have added $C_{i}$ instead of $C_{i'}$ which is a contradiction.
	If $(1+\epsilon)^{k}<2w(P(Q')\cap A)$ for the set $P(Q')$ at the beginning of
	this iteration then $w_{i}\le(1+\epsilon)^{k+1}<(1+\epsilon)2w(P(Q')\cap A)$
	and %
	\[
		\begin{aligned}
			\sum_{p\in \tilde{P}(Q) \cap C_{i}}\tilde{w}_{p}
			 & = \sum_{p\in \tilde{P}(Q) \cap C_{i}}\beta w_{p}\\
			 & \ge \beta\left(w(P(Q')\cap A)-2d\epsilon^{d+2}\tilde{W}/(d^{d+1}\log N)\right)\\
			 & \ge \beta\left(w(P(Q')\cap A)-2\epsilon^{d+2}w(P(Q'))/(d^d \log N)\right)\\
			 & \ge \beta\left(w(P(Q')\cap A)-2\epsilon w_{i}\right)\\
			 & \ge \beta\left(w_{i}/(2+2\epsilon)-2\epsilon w_{i}\right)\\
			 & = w_{i}.
		\end{aligned}
	\]
	To see that the first inequality holds, note that $w(P(Q')\cap A) \leq w(P(Q')\cap C_i) + Y$,
	where $Y$ is the weight of the points in the auxiliary grid slices of $A$ which $C_i$
	does not fully overlap. Since $A$ is the smallest aligned hyperrectangle containing
	$C_i$, in each dimension there are only two slices which are in $A$ 
	and which $C_i$ partially
	overlaps (and none which are in $A$ and do not overlap with $C_i$ at all).
	Thus, there are at most $2d$ such slices in total. Using the definition of
	the
	auxiliary grid (Lemma~\ref{lem:construct-aux-grid}), we obtain that
	$Y \leq 2d \cdot \epsilon^{d+2}\tilde{W}/(d^{d+1}\log N)$, where $\tilde{W} = w(P(Q'))$. This
	provides the first inequality. The third inequality holds because $C_i$ is
	heavy, the fourth inequality uses the above condition on $w_i$ and the last
	equality is simply the definition of $\beta$.
\end{proof}

\begin{proof}[Proof of Lemma~\ref{lem:apx-ratio-weighted-hypercubes}.]
	We can bound the weight of all light and heavy hypercubes by
   	$(2+O(\varepsilon))w(P(Q))$ by Lemmas~\ref{lem:light} and~\ref{lem:heavy}.
	Then applying Lemma~\ref{lem:point-sol} yields that
	\begin{align*}
		(2+O(\epsilon))\sum_{p\in P(Q)}w_{p}
		\le(2+O(\epsilon))\cdot2^{d+1}w(\SOL(Q)))
		=(4+O(\epsilon))2^{d}w(\SOL(Q))).
	\end{align*}

	We conclude that
	\begin{align*}
		w(\SOL(Q))\le w(\OPT(Q))\le(2+O(\epsilon))w(P(Q))\le(4+O(\epsilon))2^{d}w(\SOL(Q))).
		&\qedhere
	\end{align*}
\end{proof}

As noted before, in Appendix~\ref{sec:adjustment} we describe how
to adjust the hierarchical grid decomposition and the algorithm slightly such that
we obtain polylogarithmic update time overall. We output the solution
$\SOL:=\SOL(Q^{*})$. Using the data structure $P(Q^{*})$, in time
$O(1)$ we can also output $w(P(Q^{*}))$ which is an estimate for $w(\SOL)$
due to Lemma~\ref{lem:apx-ratio-weighted-hypercubes}.
\begin{thm}
\label{thm:weighted-hypercubes}
	For the weighted maximum independent set of hypercubes problem with weights in $[1,W]$ 
	there are fully dynamic algorithms
	that maintain $(4+O(\epsilon))2^{d}$-approximate
	solutions deterministically with worst-case update
	time~$(d/\epsilon)^{O(d^2+1/\epsilon)}\cdot\lg W\cdot\log^{2d-1}n\log^{2d+1}N$
	and with high probability with worst-case update time 
	$\left(\frac{d}{\epsilon}\right)^{O(d^{2})}\cdot\lg W\cdot\log^{2d-1}n\log^{2d+2}N$.
\end{thm}
\begin{proof}
	This follows immediately from combining Lemmas~\ref{lem:framework-dynamic}
	and \ref{lem:iterations} for the running time and
	Lemma~\ref{lem:apx-ratio-weighted-hypercubes} for the approximation ratio.
\end{proof}

\bibliographystyle{abbrv}
\bibliography{dynamic-geo-IS}
\appendix

\section{Overview of the Appendix}
\label{sec:overview-appendix}
To help the reader navigate the appendix, we provide a brief overview of the
different sections:
\begin{itemize}
	\item Section~\ref{apx:data-struct-maintain-grid}: Data Structure for Maintaining Grid Cells
	\item Section~\ref{sec:framework-details}: Details of the Algorithmic Framework
	\item Section~\ref{sec:Intervals-unweighted}: Dynamic Independent Set of Unweighted Intervals
	\item Section~\ref{sec:Unweighted-hypercubes}: Dynamic Independent Set of Unweighted Hypercubes
	\item Section~\ref{apx:weighted-intervals}: Dynamic Independent Set of Weighted Intervals
	\item Section~\ref{sec:rectangles}: Dynamic Independent Set of Rectangles
		and Hyperrectangles
	\item Section~\ref{sec:lower-bounds}: Lower Bounds
	\item Section~\ref{sec:omitted}: Omitted Proofs
\end{itemize}

\section{Data Structure for Maintaining Grid Cells}
\label{apx:data-struct-maintain-grid}
In this section we describe our data structure for maintaining the grid cells $Q$ for which $\C(Q)\ne \emptyset$. We denote by $\G'\subseteq \G$ the
set of these grid cells and we maintain them using the data structure from the following lemma. Sometimes we identify $\G'$ with this data structure.

\begin{lem}
\label{lem:grid-cells}
	There exists a data structure which maintains
	a set of grid cells $\G'\subseteq\G$ and which offers the following
	operations: 
	\begin{enumerate}
		\item Given a cell $Q\notin\G'$ we can insert $Q$ into $\G'$ and initialize
			all data structures $\D(Q)$, $\D'(Q)$, $\D_{k}(Q)$, and $\D'_{k}(Q)$
			for all $k=1,\dots,\lg_{1+\varepsilon}N$ in worst-case time $O(\lg|\G'|)$. 
		\item Given a cell $Q\in\G'$ we can remove $Q$ from $\G'$ in worst-case
			time $O(\lg|\G'|)$. 
		\item Given a cell $Q\in\G'$, we can obtain the data structures $\D(Q)$,
			$\D'(Q)$, $\D_{k}(Q)$, and $\D'_{k}(Q)$, for all $k=1,\dots,\lg_{1+\varepsilon}W$
			in worst-case time $O(\lg|\G'|)$, where $W$ is an upper bounds on
			the weights of the hypercubes in $\C$. Also, we can query or change the
			pointer to $\ALG(Q)$ in time $O(\lg|\G'|)$. 
		\item Given a cell $Q$ of level $\ell<\lg N$, we can obtain all cells
			in $\ch(Q)\cap\G'$ in worst-case time $O(2^{d}\lg|\G'|)$, 
	\end{enumerate}
\end{lem}

\begin{proof}
Our data structure works as follows. For each level $\ell=0,1,\dots,\lg N$,
we maintain an ordered list $L_{\ell}$ consisting of the non-empty
cells in $\G_{\ell}$; $L_{\ell}$ is lexicographically ordered by
the vectors $k\in\Z^{d}$ from the definition of the grid cells which uniquely
identify the cells in $\G$. Whenever a grid cell $Q\in\G$ is added,
we compute the vector $k\in\Z^{d}$ corresponding to $Q$ 
and we insert $Q$ into $L_{\ell(Q)}$. This takes worst-case time
$O(\lg|\G'|)$ since each level can contain at most $O(|\G'|)$ non-empty
cells. Similarly, when a cell $Q$ is removed then then we compute
$k$ and remove $Q$ from $L_{\ell(Q)}$ in worst-case time $O(\lg|\G'|)$.
To change a pointer for $\ALG(Q)$, we find $Q$ in $\G'$ in time
$O(\lg|\G'|)$ and then change the corresponding pointer.

For a given cell $Q$ of level $\ell<\lg N$, the hierarchical grid decomposition can identify
all $O(2^{d})$ non-empty children cells $Q_{i}\in\ch(Q)$ using the
previously defined query operation for each $Q_{i}\in\ch(Q)$. This
takes worst-case query time $O(2^{d}\lg|\G'|)$.
\end{proof}

\section{Details of the Algorithmic Framework}
\label{sec:framework-details}

In this section, we describe the formal details of the algorithmic framework
which is used by our algorithms.
 
\subsection{Algorithmic Framework for Unweighted Hypercubes}
\label{sec:framework-details:unweighted}

In this subsection, we provide the formal details of the algorithm framework
which we only sketched in Section~\ref{sec:framework}. The framework we use in
this subsection is used by our algorithms for unweighted hypercubes; its
adjustment for unweighted hypercubes is discussed in
Section~\ref{sec:adjustment}

Our algorithmic framework is a bottom-up
dynamic program. Seen as an offline algorithm, first for each cell
$Q$ of the lowest level $\lg N$ we compute a solution $\ALG(Q)\subseteq\C(Q)$
using a blackbox algorithm which we denote by \alg. Formally, $\ALG(Q):=\alg\left(Q,\DD(Q),\emptyset\right)$
where for each cell $Q$ we define $\DD(Q):=\left\{ \D(Q),\D'(Q),\left\{ \D_{k}(Q)\right\} _{k\in\N},\left\{ \D'_{k}(Q)\right\} _{k\in\N}\right\} $
to be its data structures that we described in Section~\ref{sec:framework}. The concrete implementation of \alg depends
on the dimension $d$ and whether the input hypercubes are weighted
or unweighted; we present implementations of \alg later in the paper.
Then we iterate over the levels in the order $\log N-1,\dots,0$ where
in the iteration of each level~$\ell$ we compute a solution for
the each cell $Q$ of level~$\ell$ by running \alg on $Q$. Formally,
we set $\ALG(Q):=\alg\left(Q,\DD(Q),\left\{ \ALG(Q_{i})\right\} _{Q_{i}\in c(Q)\cap\G'}\right)$
where the latter is the call to $\alg$ that obtains as input 
\begin{itemize}
\item the grid cell $Q\in\G$ together with its corresponding data structures
$\DD(Q)$ and 
\item a solution $\ALG(Q_{i})$ for each grid cell $Q_{i}\in\ch(Q)\cap\G'$,
i.e., each child of $Q$ in $\G'$. 
\end{itemize}
For each of our studied settings we will define an algorithm that
outputs a solution $\SOL$ based on the entries $\ALG(Q)$ for all
cells $Q\in\G'$. This algorithm will depend on the concrete implementation
of $\alg$ and will run in time $O(|\SOL|\log^{O(1)}N)$.

In the dynamic setting, we do the following whenever a hypercube $C$
is inserted or deleted. Let $\ell:=\ell(C)$. We first identify all
cells $Q$ for which $C\in\C(Q)$. Assume that those are the cells
$Q_{\ell},\dots,Q_{0}$ such that $\ell(Q_{j})=j$ for each $j\in\{0,\dots,\ell\}$.
We update $\G'$, i.e., insert a cell $Q_{j}$ that is not already
in $\G'$ or remove such a cell if its corresponding set $\C'(Q_{j})$
has become empty in case that we remove $C$.{} Then, we update each
data structure $\D(Q_{j})$, $\D'(Q_{j})$, $\D_{k}(Q_{j})$, $\D'_{k}(Q_{j})$
such that $C\in\C(Q_{j})$, $C\in\C'(Q_{j})$, $C\in\C_{k}(Q_{j})$,
$C\in\C'_{k}(Q_{j})$, resp., for each $j\in\{0,\dots,\ell\}$ s.t.\ $Q_{j}\in\G'$.

We recompute the solutions $\ALG(Q_{j})$ in decreasing order of $j=\ell,\dots,0$
in the same way as described above, i.e., we define $\ALG(Q_{j}):=\alg\left(Q_{j},\DD(Q_{j}),\emptyset\right)$
if $\ell(Q_{j})=\log N$ and $\ALG(Q_{j}):=\alg\left(Q_{j},\DD(Q_{j}),\left\{ \ALG(Q_{j'})\right\} _{Q_{j'}\in\ch(Q_{j})\cap\G'}\right)$.
This requires $O(\lg N)$ calls to the black box algorithm \alg and
takes a total time of $O(2^{d}\lg n\lg N)$ to identify the children
cells $c(Q_{j})$ for all $Q_{j}$ with $j=\ell,\dots,0$. Thus, the
algorithmic framework has the following guarantee on the update times. 
\begin{lem}
\label{lem:framework-dynamic}
	Assume that the algorithm $\alg$ runs
	in time $T$. Then we obtain a worst-case update time of $(d/\epsilon)^{O(1/\epsilon)}\lg N(\lg^{2d+1}n+2^{d}\lg n+T))$
	in the deterministic setting and $O(\lg^{2}N(\lg^{2d+1}n+2^{d}\lg n+T))$
	in the randomized setting. 
\end{lem}
\begin{proof}
	This follows immediately from the discussion preceding
	Lemma~\ref{lem:framework-dynamic} together with Lemma~\ref{lem:offset} for the number of
	offsets, Lemma~\ref{lem:data-structure} for the insertion and deletion times
	of hypercubes and Lemma~\ref{lem:grid-cells} for the updates of the grid $\G'$.
\end{proof}

\subsection{Adjustment of Hierarchical Grid Decomposition for Weighted Hypercubes}
\label{sec:framework-details:weighted}
\label{sec:adjustment}

Our implementation of the algorithm from Section~\ref{sec:Weighted-hypercubes} might
not run in time $\polylg(n, N)$ since initializing the data structure
$P(Q)$ takes time $\Omega(|P(Q)|)$ and it is possible that $|P(Q)| = \Omega(n)$.
Note, however, that initializing the data structure $P(Q)$ is the only operation
in our algorithm for weighted hypercubes which potentially requires
superpolylogarithmic time.  We describe a slight modification of our
general hierarchical grid decomposition such that we can maintain $P(Q)$ efficiently and obtain a total
update times of $\polylg(n, N)$. 

Roughly speaking, our approach works as follows. Instead of building $P(Q)$ from
scratch at the beginning of each update, we update it in the background whenever
necessary. In particular, whenever a cell $Q'$ adds/removes a hypercube $C$ to
$\bar{\C}(Q')$, then the vertices of $C$ are added/removed from all data
structures $P(Q)$ such that $Q' \subset Q$; note that this effects at most $\lg N$
cells in total.  This will ensure that when running the algorithm for
a cell $Q$, then at the beginning the set $P(Q)$ already equals
$\bigcup_{Q'\in\ch(Q)}P(Q')$.

We now describe this formally. First, let us introduce several additional data
structures. We introduce an additional global data structure
$\bar{\D}$ according to Lemma~\ref{lem:data-structure} that maintains the set
$\bar{\C}:=\bigcup_{Q\in\G'}\bar{\C}(Q)$, i.e., $\bar{D}$ contains the
hypercubes $\bar{\C}(Q)$ for all non-empty cells $Q\in\G'$ (recall that $\G'$ is
the set of non-empty grid cells).  Also, for each cell $Q\in\G'$ we introduce a
data structure $\tilde{\D}(Q)$ that stores a set $\tilde{\C}(Q)\subseteq\C'(Q)$
that contains all hypercubes contained in the set $\bar{\C}(Q)$ at some point
during the execution of the algorithm,
i.e., $\tilde{C}(Q)$ also contains hypercubes which (during a single run of the
algorithm) are inserted into $\bar{\C}$ and later removed.
Note that since the algorithm runs for at most $d^d \log W/\epsilon^{d}$ iterations
(see Lemma~\ref{lem:iterations}) and each iteration inserts at most one
hypercube into $\bar{\C}(Q)$, we have that $|\tilde{\C}(Q)|\le d^d \log W/\epsilon^{d}$.  
Furthermore, we adjust the data structure due to Lemma~\ref{lem:grid-cells} such
that we can access $\tilde{\D}(Q)$ in time $O(\log|\G'|)$ (like
$\D(Q)$ and $\D'(Q),$ etc.).

We explain now how to adjust our hierarchical grid decomposition in case that a hypercube is inserted or deleted.
As before, when a hypercube $C$ is inserted or deleted, we update
the data structure maintaining the non-empty grid cells $\G'$ and we update the
data structures $\D(Q)$, $\D'(Q)$, $\D_{k}(Q)$, and $\D'_{k}(Q)$ for all grid
cells $Q$ and all $k\in\N$ such that $C\in\C(Q)$, $C\in\C'(Q)$, $C\in\C_{k}(Q)$,
$C\in\C'_{k}(Q)$, respectively. As before, this takes time $O(\lg N(\lg n+\lg^{2d+1}n))$.
Let $Q_{j},\dots,Q_{0}$ be all cells $Q$ such that $C\in\C(Q)$ and
assume that $\ell(Q_{j'})=j'$ for each $j'\in\{0,\dots,j\}$. Recall
that in each cell $\ALG(Q_{j})$ we store $(\bar{\C}(Q_{j}),P(Q_{j}))$
where $|\bar{\C}(Q_{j})|\le (d/\epsilon)^{d}$. For each $j'\in\{0,\dots,j\}$
we delete all hypercubes from $\bar{\C}(Q_{j'})$ and
$\tilde{\C}(Q_{j'})$ since we will recompute all these sets and we
remove the hypercubes in $\bar{\C}(Q_{j'})$ from $\bar{\D}$. Also,
we update the set $P(Q_{j''})$ for each $j''\le j'$ accordingly
such that it no longer contains the weights of the vertices of the
hypercubes in $\tilde{\C}(Q_{j'})$. More formally, for each $j'\in\{0,\dots,j\}$,
each $C_{i}\in\tilde{\C}(Q_{j'})$, each $j''\in\{0,\dots,j'\}$, and
each vertex $p$ of $C_{i}$ we decrement the weight of $p$ in $P(Q_{j''})$
by $w_{i}$. After that, for each $j'\in\{0,\dots,j\}$ and for each
$C_{i}\in\bar{\C}(Q_{j'})$ we remove $C_{i}$ from $\bar{\D}$ and
finally set $\bar{\C}(Q_{j'}),\tilde{\C}(Q_{j'})=\emptyset$.
This takes time $O(d^d \log^{2d}n \lg W / \varepsilon^d)$ since the most expensive
operation is to delete the $O(d^d \lg W / \varepsilon^d)$ hypercubes from $\tilde{\C}(Q)$
where each deletion takes time $O(\lg^{2d} n)$; note that this time is later
subsumed by the running time of the algorithm from Lemma~\ref{lem:iterations}.
Finally, we call the algorithm
for each cell $Q$ in $Q_{j},Q_{j-1},\dots,Q_{0}$ in this order.

Now let us look how we need to adapt the algorithm (as defined in
Section~\ref{sec:Weighted-hypercubes}) to obtain polylogarithmic running
times. When running the algorithm, we omit the step where we define $P(Q)$ to be
$\bigcup_{Q'\in\ch(Q)}P(Q')$. Whenever we add a hypercube
$C_{i}$ to $\bar{\C}(Q)$ then we add $C_{i}$ also to $\tilde{\C}(Q)$
and for each cell $Q'$ with $Q\subseteq Q'$, we increase in $P(Q')$
the weight of all vertices of $C_{i}$ by $w_{i}$. This ensures that
when we run the algorithm for some cell $Q$ then at the beginning the set
$P(Q)$ already equals $\bigcup_{Q'\in\ch(Q)}P(Q')$. Observe that,
hence, the algorithm directly changes entries $\ALG(Q')$ for cells $Q'\ne Q$
which we did not allow before in our hierarchical grid decomposition. Also, whenever we add a hypercube $C_{i}$
to $\bar{\C}(Q)$ then we add $C_{i}$ also to $\bar{\D}$, and whenever
we remove a hypercube $C_{i}$ from $\bar{\C}(Q)$ then we remove
$C_{i}$ also from $\bar{\D}$.

Finally, we devise the routine for returning the global solution $\SOL(Q^{*})$
in $\tilde{O}(|\SOL(Q^{*})|)$ time using a recursion on the grid cells that
stops if a grid cell $Q$ does not contain any hypercube from $\SOL(Q^{*})$. 
\begin{lem}
\label{lem:output-routine-weighted-hypercubes}
In time $|\SOL(Q^{*})|\cdot (d/\epsilon)^{O(d^{2})} \cdot \log^{d+1}N$,
we can output all hypercubes in $\SOL(Q^{*})$. 
\end{lem}
\begin{proof}
Our output routine is a recursive algorithm which is first called on $Q^*$. The
input of each recursive call is a cell $Q$ and a set of at most $(d/\epsilon)^{d}
\cdot (\ell(Q)-1)$ hypercubes $\hat{\C}$; $\hat{\C}$ contains those hypercubes
in $\SOL(Q^{*})$ that intersect $Q$ and which originate from levels higher than
$Q$.  Now we first check
whether there is a hypercube 
$C\in\bar{D}$ 
with $C\subseteq Q$ such that
$C\cap C'=\emptyset$ for each $C'\in\hat{\C}$. Using $\bar{\D}$ we can do this
in time
\begin{align*}
	O\left(\left( (d/\epsilon)^{d} \cdot (\ell(Q)-1)\right)^{d}\log^{2d-1}|\bar{\C}(Q)|\right)
	&= O\left( (d/\epsilon)^{d^{2}}\log^{d}N \lg^{2d-1}\left((d/\varepsilon)^d\right)\right) \\
	&=(d/\epsilon)^{O(d^{2})}\log^{d}N
\end{align*}
using an auxiliary grid similarly as in the case of unweighted hypercubes.
(Note that we cannot simply go through all hypercubes in $\bar{\C}(Q)$ and then 
recurse on each child of $Q$ since then our running time might not be near-linear in $|\SOL(Q^{*})|$.)
If not then we stop. Otherwise, we output all hypercubes in $\bar{\C}(Q)$
that do not intersect any hypercube in $\hat{\C}$ and recurse on
each child of $Q$ that contains a hypercube that does not intersect any hypercube in 
$\bar{\C}(Q)\cup \hat{\C}$ (which we check again using an auxiliary grid).
The argument for each recursive call is $\bar{\C}(Q)\cup \hat{\C}$. 
The recursion tree has at most $O(|\SOL(Q^{*})|\log N)$ nodes since
we stop if a given cell $Q$ does not contain a hypercube $C\in\bar{\C}$
with $C\cap C'=\emptyset$ for each $C'\in\hat{\C}$. Hence, this
algorithm has running time
$|\SOL(Q^{*})|\cdot (d/\epsilon)^{O(d^{2})} \cdot \log^{d+1}N$
overall.%
\end{proof}

\section{Dynamic Independent Set of Unweighted Intervals}
\label{sec:Intervals-unweighted}

We describe our dynamic algorithm for unweighted intervals, i.e., we assume that
$d=1$ and $w_{i}=1$ for each interval $C_{i}\in\C$.  We prove the following
theorem.
\begin{thm}
	\label{thm:unweighted-intervals}
	For the unweighted maximum independent set of intervals problem there are
	fully dynamic algorithms that maintain $(1+\epsilon)$-approximate solutions
	deterministically with worst-case update time
	$(1/\epsilon)^{O(1/\epsilon)} \left(\frac{1}{\varepsilon}\lg n\log N\right)^{2}$
	and with high probability with worst-case update time 
	$O\left(\frac{1}{\varepsilon^2} \log^3 N \log^2 n \right)$. 
\end{thm}

Roughly speaking, our algorithm works as follows. If for a cell $Q$,
$\left|\bigcup_{Q'\in\ch(Q)}\ALG(Q')\right|$ is small, we compute the
\emph{optimal} solution for $\C(Q)$ (without using the solutions $\ALG(Q')$ for
the cells $Q'\in\ch(Q)$) in time $O(\log N\log^{2}n/\epsilon^{2})$, using the
data structure $\D(Q)$. If $\left|\bigcup_{Q'\in\ch(Q)}\ALG(Q')\right|$ is
large, then we output $\bigcup_{Q'\in\ch(Q)}\ALG(Q')$ as above. Using a charging
argument, we argue that we lose at most a factor of $1+\epsilon$ by ignoring the
intervals in $\C'(Q)$ in the latter case.

In the following, for a cell $Q$ we write $\OPT_{U}(Q)\subseteq\C(Q)$ to denote
the independent set of the intervals in $\C(Q)$ of maximum cardinality. To
simplify our notation, for each interval $C_{i}=(x_{i}^{(1)},y_{i}^{(1)})\in\C$
we define $x_{i}:=x_{i}^{(1)}$ and $y_{i}:=y_{i}^{(1)}$.

We describe the concrete implementation of the algorithm, where we assume that
the algorithm
obtains as input a cell $Q\in\G$ with children $Q_{1},Q_{2}$, the
set of data structures $\DD(Q):=\left\{ \D(Q),\D'(Q),\left\{ \D_{k}(Q)\right\} _{k\in\N}\right\} $,
and solutions $\ALG(Q_{1}),\ALG(Q_{2})$ for the two children cells.
First, we check if $|\ALG(Q_{1})\cup\ALG(Q_{2})|>\frac{1}{\epsilon^{2}}\log N$.
If this is the case, then the algorithm returns $\ALG(Q_{1})\cup\ALG(Q_{2})$
via pointers to $\ALG(Q_{1})$ and $\ALG(Q_{2})$, i.e., no further hypercubes are
added to the solution. Otherwise, the algorithm
uses a subroutine $\SUB$ which we define in the next lemma to compute
an independent set $\SUB(Q)\subseteq\C(Q)$ with $|\SUB(Q)|=|\OPT_{U}(Q)|$;
then the algorithm returns $\SUB(Q)$.
\begin{lem}
\label{lem:subrut-unweighted} Given a cell $Q$ we can compute a
solution $\SUB(Q)\subseteq\C(Q)$ with $|\SUB(Q)|=|\OPT_{U}(Q)|$
in time $O(|\OPT_{U}(Q)|\cdot\log^{2}n)$. 
\end{lem}

\begin{proof}
The subroutine works by running the following greedy algorithm which
is known to compute a maximum cardinality independent set (see, e.g.,
Agarwal et al.~\cite{agarwal98label}). Let $Q=[a,b)$ be the current
cell and let $\SUB(Q)=\emptyset$. Start by finding the interval $C_{1}=[x_{1},y_{1})\in\C(Q)$
with the smallest $y_{1}$ coordinate and add $C_{1}$ to $\SUB(Q)$,
using the fourth property of Lemma~\ref{lem:data-structure} for
the data structure $\D(Q)$ with $[t,t')=[a,b)$. Now repeatedly find
the interval $C_{i}=[x_{i},y_{i})\in\C(Q)$ with the smallest $y_{i}$
coordinate such that $x_{i}\geq y_{i-1}$, using the fourth property
of Lemma~\ref{lem:data-structure} with $\D(Q)$ and $[t,t')=[y_{i-1},b)$.
We repeat this procedure until no such $C_{i}$ exists. Then, the
total running time of the algorithm is $O(|\OPT_{U}(Q)|\cdot\lg^{2}n)$.%
\end{proof}
We output the solution $\SOL:=\ALG(Q^{*})$ as follows: if $\ALG(Q^{*})$
contains a list if intervals, then we output those. Otherwise $\ALG(Q^{*})$
contains pointers to two solutions $\ALG(Q_{1}),\ALG(Q_{2})$ and
we recursively output those. Like in the case of unweighted hypercubes, 
for each cell $Q$ in $\ALG(Q)$ we store additionally $|\ALG(Q)|$
and hence we can report $|\SOL|=|\ALG(Q^{*})|$ in time $O(1)$.

\subsection{Analysis}

We analyze the running time and approximation ratio of the algorithm. 
Its running time is $O(\frac{1}{\varepsilon^{2}}\lg^{2}n\lg N)$ since
if $|\ALG(Q_{1})\cup\ALG(Q_{2})|>\frac{1}{\epsilon^{2}}\lg N$ we
return $\ALG(Q_{1})\cup\ALG(Q_{2})$ in time $O(1)$, and otherwise
our computed solution has size $O\left(\frac{1}{\varepsilon^{2}}\lg N\right)$,
using that the maximum cardinality independent set in $\C'(Q)$ has
size at most $1/\epsilon$ (see Lemma~\ref{lem:properties-grid})
and that $|\ALG(Q_{1})|=|\OPT_{U}(Q_{1})|$ and $|\ALG(Q_{2})|=|\OPT_{U}(Q_{2})|$
in this case. 
\begin{lem}
\label{lem:intervals-unweighted-time} The worst-case running time
of the algorithm is $O(\frac{1}{\varepsilon^{2}}\lg^{2}n\lg N)$. 
\end{lem}

\begin{proof}
First, the algorithm checks if $|\ALG(Q_{1})\cup\ALG(Q_{2})|>\frac{1}{\epsilon^{2}}\lg N$.
We can do this in time $O(\frac{1}{\varepsilon^{2}}\lg^{2}n\lg N)$
by recursively outputting $\ALG(Q_{1})$ and $\ALG(Q_{2})$ and stopping
after outputting $\frac{1}{\epsilon^{2}}\lg N+1$ intervals. If $|\ALG(Q_{1})\cup\ALG(Q_{2})|>\frac{1}{\epsilon^{2}}\lg N$
then the algorithm needs only time $O(1)$ to output $\ALG(Q_{1})\cup\ALG(Q_{2})$.
If $|\ALG(Q_{1})\cup\ALG(Q_{2})|\le\frac{1}{\epsilon^{2}}\lg N$, we
only need to bound the running time of the subroutine $\SUB$ from
Lemma~\ref{lem:subrut-unweighted}. To do this, we show have to show
that $|\OPT_{U}(Q)|=O(\frac{1}{\varepsilon^{2}}\lg N)$ which implies
the lemma.

To bound $|\OPT_{U}(Q)|$, note that since $|\ALG(Q_{1})\cup\ALG(Q_{2})|\leq\frac{1}{\epsilon^{2}}\lg N$,
we must have that $|\ALG(Q_{1})|=|\OPT_{U}(Q_{1})|$ and $|\ALG(Q_{2})|=|\OPT_{U}(Q_{2})|$
by the definition of the algorithm. The cardinality of $\OPT_{U}(Q)$ is at
most $|\OPT_{U}(Q_{1})|+|\OPT_{U}(Q_{2})|$ plus the maximum cardinality
of an independent set in $\C'(Q)$. The latter quantity is bounded
by $1/\varepsilon$ due to Lemma~\ref{lem:properties-grid}. Hence,
we obtain that

\begin{eqnarray*}
|\OPT_{U}(Q)| & \leq & |\OPT_{U}(Q)\cap\C(Q_{1})|+|\OPT_{U}(Q)\cap\C(Q_{2})|+|\OPT_{U}(Q)\cap\C'(Q)|\\
 & \le & |\OPT_{U}(Q_{1})|+|\OPT_{U}(Q_{2})|+|\OPT_{U}(Q)\cap\C'(Q)|\\
 & \le & \frac{1}{\varepsilon^{2}}\lg N+\frac{1}{\varepsilon}\\
 & = & O\left(\frac{1}{\varepsilon^{2}}\lg N\right).
\end{eqnarray*}
\end{proof}
Now we show that the global solution computed by the algorithm indeed is a
$(1+O(\varepsilon))$-approximation of $\OPT_{U}$. 
\begin{lem}
\label{lem:intervals-unweighted-approximation} The returned solution
$\SOL$ is a $(1+O(\varepsilon))$-approximation of $\OPT_{U}$. 
\end{lem}

\begin{proof}
We use a charging argument to analyze the algorithm. Let $Q$ be any cell
in $\G$ with children $\ch(Q)=\{Q_{1},Q_{2}\}$. First, if $|\ALG(Q_{1})\cup\ALG(Q_{2})|\le\frac{1}{\epsilon^{2}}\log N$
then $|\SUB(Q)|=|\OPT_{U}(Q)|$ by the definition of the algorithm. Second,
assume that $|\ALG(Q_{1})\cup\ALG(Q_{2})|>\frac{1}{\epsilon^{2}}\log N$.
Note that in this case $\OPT_{U}(Q)$ might contain some intervals
$C_{i}\in\C'(Q)$ which the algorithm does not pick. However, $\OPT_{U}(Q)$
can contain at most $1/\epsilon$ such intervals by Lemma~\ref{lem:properties-grid}.
We charge the intervals in $\OPT_{U}(Q)\cap\C'(Q)$ to the intervals
in $\ALG(Q_{1})\cup\ALG(Q_{2})$ (which are selected by the algorithm). Thus,
each interval in $\ALG(Q_{1})\cup\ALG(Q_{2})$ receives a charge of
\begin{align*}
\frac{1/\varepsilon}{|\ALG(Q_{1})\cup\ALG(Q_{2})|}\leq\frac{\varepsilon}{\lg N}.
\end{align*}

Now consider the solution $\SOL$ returned by the algorithm. Since each interval
is contained in at most $\log N$ cells of $\G$ by Lemma~\ref{lem:properties-grid},
each interval in $\SOL$ is charged at most $\lg N$ times. Hence,
each interval in $\SOL$ receives a total charge of at most $\varepsilon$
and, thus, $\SOL$ satisfies $|\SOL|\geq(1-\varepsilon)|\OPT_{U}|$. 
\end{proof}
By combining Lemma~\ref{lem:intervals-unweighted-time}, Lemma~\ref{lem:intervals-unweighted-approximation}
and Lemma~\ref{lem:framework-dynamic}, we complete the proof of
Theorem~\ref{thm:unweighted-intervals}.

\section{Unweighted Hypercubes}
\label{sec:Unweighted-hypercubes}

In this section we assume that $w_{i}=1$ for each $C_{i}\in\C$ and we present a
dynamic algorithm for hypercubes with an approximation ratio of
$(1+\epsilon)2^{d}$. Our result is stated in the following theorem.
\begin{thm}
\label{thm:unweighted-hypercubes}
	For the unweighted maximum independent set of hypercubes problem in $d$
	dimensions there are fully dynamic algorithms that maintain
	$(1+\epsilon)2^{d}$-approximate solutions deterministically with worst-case
	update time $(d/\epsilon)^{O(d^2+1/\epsilon)} \lg^{2d+1}N\lg^{2d+1}n$ and
	with high probability with worst-case update time $(d/\epsilon)^{O(d^2)}
	\lg^{2d+2}N\lg^{2d+1}n$. 
\end{thm}

We define the algorithm for hypercubes.
Let $x_{Q}^{(1)},\dots,x_{Q}^{(d)}$ and $y_{Q}^{(1)},\dots,y_{Q}^{(d)}$
be such that $Q=[x_{Q}^{(1)},y_{Q}^{(1)}]\times\dots\times[x_{Q}^{(d)},y_{Q}^{(d)}]$.
If $\left|\bigcup_{Q'\in\ch(Q)}\ALG(Q')\right|>d^d \log N/\epsilon^{d+1}$,
we output $\bigcup_{Q'\in\ch(Q)}\ALG(Q')$ by returning a list of pointers to
$\{\ALG(Q')\}_{Q'\in \ch(Q)}$. Otherwise, roughly speaking
we sort the hypercubes $C\in\C'(Q)$ non-decreasingly by size and
add them greedily to our solution, i.e., we add a hypercube $C\in\C'(Q)$
if $C$ does not overlap with a hypercube in $\bigcup_{Q'\in\ch(Q)}\ALG(Q')$
or with a previously selected hypercube from $\C'(Q)$. To this end,
we first initialize our solution for $Q$ to $\bigcup_{Q'\in\ch(Q)}\ALG(Q')$.
Then, we proceed in iterations where in each iteration we add a hypercube
from $\C'(Q)$ to the solution. %
Suppose that at the beginning of the current iteration, the current
solution is $\bar{\C}\subseteq\C(Q)$, containing all hypercubes in
$\bigcup_{Q'\in\ch(Q)}\ALG(Q')$ and additionally all hypercubes from
$\C'(Q)$ that we selected in previous iterations. We want to find
the hypercube $C_{i^{*}}\in\C'(Q)$ with smallest size $s_{i^{*}}$ which has the
property
that it does not overlap with any hypercube in $\bar{\C}$. To this end,
we construct an \emph{auxiliary grid} inside $Q$ with coordinates
$Z^{(j)}:=\{x_{i}^{(j)},y_{i}^{(j)}|C_{i}\in\bar{\C}\}\cup\{x_{Q}^{(j)},y_{Q}^{(j)}\}$
for each dimension $j\in[d]$. Observe
that if such a $C_{i^{*}}$ exists, it overlaps with some set of auxiliary grid cells;
let $B^{*}$ denote the union of these cells. Then $B^{*}$
does not intersect with any hypercube in $\bar{\C}$ and is \emph{aligned
with} $Z^{(1)},\dots,Z^{(d)}$, where we say that a set $B\subseteq\mathbb{R}^{d}$
is aligned with $Z^{(1)},\dots,Z^{(d)}$ if there are values $x^{(j)},y^{(j)}\in Z^{(j)}$
for each $j\in[d]$ such that $B=(x^{(1)},y^{(1)})\times\dots\times(x^{(d)},y^{(d)})$. 
\begin{lem}
\label{lem:aligned-hyperrectangle}
	If $C_{i^*}$ exists, then there exists a hyperrectangle $B^{*}\subseteq Q$ such
	that $C_{i^{*}}\subseteq B^{*}$, $B^{*}$ is aligned with $Z^{(1)},\dots,Z^{(d)}$
	and that satisfies $B^{*}\cap C_{i}=\emptyset$ for each $C_{i}\in\bar{\C}$. 
\end{lem}
\begin{proof}
	Let $C_{i^{*}}=(x_{i^{*}}^{(1)},y_{i^{*}}^{(1)})\times\dots\times(x_{i^{*}}^{(d)},y_{i^{*}}^{(d)})$.
	For all $j\in[d]$, set $x^{(j)}\in Z^{(j)}$ to the largest coordinate
	in $Z^{(j)}$ such that $x^{(j)}\leq x_{i^{*}}^{(j)}$ and set $y^{(j)}\in Z^{(j)}$
	to the smallest coordinate in $Z^{(j)}$ such that $y^{(j)}\geq y_{i^{*}}^{(j)}$.
	Now let $B^{*}=(x^{(1)},y^{(1)})\times\dots\times(x^{(d)},y^{(d)})$.
	Clearly, $B^{*}$ is aligned with $Z^{(1)},\dots,Z^{(d)}$.

	We prove the last property claimed in the lemma by contradiction.
	Suppose there exists a $C_{i}=\prod_{j=1}^{d}(x_{i}^{(j)},y_{i}^{(j)})\in\bar{\C}$
	such that $B^{*}\cap C_{i}\neq\emptyset$. Then there exists a coordinate
	$j\in[d]$ such that $x^{(j)}< x_{i}^{(j)}\leq x_{i^{*}}^{(j)}$
	or $y_{i^{*}}^{(j)}\leq y_{i}^{(j)}< y^{(j)}$. However, the definition
	of $Z^{(j)}$ implies that $x_{i}^{(j)}\in Z^{(j)}$ and $y_{i}^{(j)}\in Z^{(j)}$.
	This contradicts our above choices of $x^{(j)}$ and $y^{(j)}$ which
	were picked as the largest (smallest) coordinates in $Z^{(j)}$ such
	that $x^{(j)}\leq x_{i^{*}}^{(j)}$ ($y^{(j)}\geq y_{i^{*}}^{(j)}$).
\end{proof}

We enumerate all possibilities for $B^{*}$. We discard a candidate for $B^{*}$
if there is a hypercube $C_{i}\in\bar{\C}$
with \emph{$B^{*}\cap C_{i}\ne\emptyset$}; this is done by iterating
over all $C_{i}\in\bar{\C}$ and checking whether $B^{*}\cap C_{i}\neq\emptyset$.
Note that there are only $O(\prod_{j=1}^{d}|Z^{(j)}|^{2})=O\left((2|\bar{\C}|)^{2d}\right)$
possibilities for $B^{*}$
that are aligned with $Z^{(1)},\dots,Z^{(d)}$. For each such possibility
for $B^{*}$, we compute the hypercube $C_{i}\in\C'(Q)$ with smallest
size $s_{i}$ such that $C_{i}\subseteq B^{*}$, using the data structure
$\D'(Q)$ from Lemma~\ref{lem:data-structure}. Let $C_{i^{*}}$
be the hypercube with smallest size among the hypercubes that we found
for all candidates for $B^{*}$. We select $C_{i^{*}}$ and add it
to $\bar{\C}$. This completes one iteration. We stop if in one iteration
we do not find a hypercube $C_{i^{*}}$ that we can add to $\bar{\C}$.
Due to Lemma~\ref{lem:properties-grid} any feasible solution can
contain at most $(d/\epsilon)^{d}$ hypercube from $\C'(Q)$. Hence,
the number of iterations is at most $(d/\epsilon)^{d}$. Finally,
let $\bar{\C}$ denote the union of all hypercubes in $\bigcup_{Q'\in\ch(Q)}\ALG(Q')$
and all hypercubes from $\C'(Q)$ that we selected in some iteration
of our algorithm. Then the algorithm returns $\bar{\C}$ as a list of hypercubes. 

\begin{lem}
\label{lem:subroutine-unweighted-hypercubes}
	We have $|\bar{\C}|=O(\frac{d^d}{\varepsilon^{d+1}}\lg N)$ and computing
	$\bar{\C}$ takes a total time of
	$O(\left(\frac{2d}{\epsilon}\right)^{d(d+1)}\log^{2d}N\lg^{2d+1}n)$. 
\end{lem}
\begin{proof}
	Before the first iteration starts, we have that
	$|\bar{\C}|=|\bigcup_{Q'\in\ch(Q)}\ALG(Q)|\leq\frac{d^d}{\varepsilon^{d+1}}\lg N$. 
   	In each iteration we add at most one hypercube from $\C'(Q)$ and, by
	Lemma~\ref{lem:properties-grid}, at most $(d/\varepsilon)^{d}$ such hypercubes
	can be added to $\bar{\C}$. Thus, when we stop
	$|\bar{\C}| \leq \frac{d^d}{\varepsilon^{d+1}}\lg N + (d/\varepsilon)^{d} =
	O(\frac{d^d}{\varepsilon^{d+1}}\lg N)$.

	Now let us bound the running time of the algorithm. There are
	$O((d/\varepsilon)^{d})$ iterations for adding hypercubes from $\C'(Q)$. Each
	iteration performs
	$O\left((2|\bar{\C}|)^{2d}\right)=O(\left(\frac{2d}{\epsilon}\right)^{d(d+1)}\log^{2d}N)$
	guesses for $B^{*}$. In each iteration, we can check if $B^{*}\cap
	C_{i}\neq\emptyset$ for all $C_{i}\in\bar{\C}$ in time $O(|\bar{\C}|d)$ as
	follows.  Given $C_{i}=\prod_{j=1}^{d}[x_{i}^{(j)},y_{i}^{(j)}]\in\bar{\C}$
	and $B^{*}=\prod_{j=1}^{d}(x_{B^{*}}^{(j)},y_{B^{*}}^{(j)})$, the algorithm
	checks if
	$(x_{B^{*}}^{(j)},y_{B^{*}}^{(j)})\cap(x_{i}^{(j)},y_{i}^{(j)})\neq\emptyset$
	for all $j=1,\dots,d$; this takes time $O(d)$. After that, the algorithm
	spends time $O(\lg^{2d+1}n)$ to find the smallest hypercube in $\C'(Q)$
	using the data structure from Lemma~\ref{lem:data-structure}.
	Since $O(|\bar{\C}|)=O(\frac{d^d}{\varepsilon^{d+1}}\lg N)$, the total running
	time for the procedure is
	$O(\left(\frac{2d}{\epsilon}\right)^{d(d+1)}\log^{2d}N\lg^{2d+1}n)$.
\end{proof}

It remains to bound the approximation ratio of the overall algorithm.
First, assume that we are in the case that
$\left|\bigcup_{Q'\in\ch(Q)}\ALG(Q')\right|\leq\frac{d^d}{\epsilon^{d+1}}\log N$.
We show that then $|\bar{\C}|\le2^{d}\cdot|\OPT(Q)|$, where $\OPT(Q)$
denotes the optimal solution for the hypercubes $\C(Q)$. We use a
charging argument. When our algorithm selects a hypercube $C_{i}$
with $C_{i}\notin\OPT(Q)$ then potentially there can be a hypercube
$C_{i'}\in\OPT(Q)$ that cannot be selected later because $C_{i}\cap C_{i'}\ne\emptyset$.
However, our algorithm selects the input hypercubes greedily, ordered
by size. Therefore, we show that for each selected hypercube $C_{i}\in\bar{\C}$
there are at most $2^{d}$ hypercubes $C_{i'}\in\OPT(Q)$ that appear
\emph{after} $C_{i}$ in the ordering such that $C_{i}\cap C_{i'}\ne\emptyset$,
i.e., such that $C_{i}$ prevented the algorithm from selecting $C_{i'}$.
In such a case the hypercube $C_{i'}$ must overlap a vertex of $C_{i}$
and since $C_{i}$ has only $2^{d}$ vertices this yields the approximation
ratio of $2^{d}$. 
\begin{lem}
\label{lem:charging-2}
	Let $C_{i}\in\C(Q)$. There are at most $2^{d}$ hypercubes $C_{i'}\in\OPT(Q)$
	s.t.\ $C_{i}\cap C_{i'}\ne\emptyset$ and $s_{i'}\geq s_{i}$. 
\end{lem}
\begin{proof}
	Let $C_{i'}\in\OPT(Q)$ be a hypercube with the properties claimed
	in the lemma. We show that $C_{i'}$ must overlap with at least one
	vertex of $C_{i}$. Since $C_{i}$ has $2^{d}$ vertices, this implies
	the lemma. Since $C_{i}\cap C_{i'}\neq\emptyset$ and $s_{i'}\geq s_{i}$,
	for all dimensions $j\in[d]$ we must have that either $x_{i}^{(j)}\leq x_{i'}^{(j)}\leq y_{i}^{(j)}\leq y_{i'}^{(j)}$
	or $x_{i'}^{(j)}\leq x_{i}^{(j)}\leq y_{i'}^{(j)}\leq y_{i}^{(j)}$
	or $x_{i'}^{(j)}\leq x_{i}^{(j)}\leq y_{i}^{(j)}\leq y_{i'}^{(j)}$.
	 This implies that in each dimension $j\in[d]$ there exists a point
	$p_{j}\in\{x_{i}^{(j)},y_{i}^{(j)}\}$ such that $p_{j}\in[x_{i'}^{(j)},y_{i'}^{(j)}]$.
	Now observe that the point $p=(p_{1},\dots,p_{d})\in\R^{d}$ is a
	vertex of $C_{i}$ and $p\in C_{i'}$. This implies that $C_{i'}$
	overlaps with at least one vertex of $C_{i}$ and since the hypercubes in
	$\OPT(Q)$ are non-intersecting, $C_{i'}$ is the only hypercube of $\OPT(Q)$ that
	overlaps with this vertex of $C_i$.
\end{proof}

For each $C_{i}\in\bar{\C}$ we charge the at most $2^{d}$ hypercubes $C_{i'}$
due to Lemma~\ref{lem:charging-2} to $C_{i}$ which proves an approximation ratio
of $2^{d}$ for the case that
$\left|\bigcup_{Q'\in\ch(Q)}\ALG(Q')\right|\leq\frac{d^d}{\epsilon^{d+1}}\log N$.
If $\left|\bigcup_{Q'\in\ch(Q)}\ALG(Q')\right|>\frac{d^d}{\epsilon^{d+1}}\log N$
then we charge the at most $(d/\epsilon)^{d}$ hypercubes in $\OPT(Q)\cap\C'(Q)$
to the hypercubes in $\bigcup_{Q'\in\ch(Q)}\ALG(Q')$. Each of the latter
hypercubes receives a charge of at most $\epsilon / \log N$ in this level, and
since there are $\log N$ levels, it receives a charge of at most $\epsilon$ in
total.  Hence, we obtain an approximation ratio of $(1+\epsilon)2^{d}$.

We define $\SOL:=\ALG(Q^{*})$ to be our computed solution and we output $\SOL$
as follows: if $\ALG(Q^{*})$ contains a list if hypercubes, then we output
those.  Otherwise $\ALG(Q^{*})$ contains pointers to solutions $\ALG(Q')$ with
$Q' \in  \ch(Q)$ and we recursively output these solutions. Hence, we can output
$\SOL$ in time $O(|\SOL|)$.  Finally, we can store each solution $\ALG(Q)$ such
that it has one entry storing its cardinality $|\ALG(Q)|$. We can easily
recompute $|\ALG(Q)|$ whenever our algorithm above recomputes an entry
$\ALG(Q)$.  Then, in time $O(1)$ we can report the size of $\SOL$ by simply
returning $|\ALG(Q^{*})|$.

We conclude the section with the proof of
Theorem~\ref{thm:unweighted-hypercubes}.
\begin{proof}[Proof of Theorem~\ref{thm:unweighted-hypercubes}]
	This follows immediately from combining Lemmas~\ref{lem:framework-dynamic}
	and \ref{lem:subroutine-unweighted-hypercubes} for the running time and the
	charging argument after Lemma~\ref{lem:charging-2} for the approximation
	ratio.
\end{proof}

\section{Dynamic Independent Set of Weighted Intervals}
\label{apx:weighted-intervals}

We present a $(1+\epsilon)$-approximation algorithm for fully dynamic
IS of weighted intervals with arbitrary weights.  Our main result
is summarized in the following theorem.
\begin{thm}
\label{thm:weighted-intervals-1}
	For the weighted maximum IS of intervals problem 
	there exists a fully dynamic algorithm which
	maintains $(1+\epsilon)$-approximate solution deterministically with worst-case
	update time~$\updateweightedintervals$.
\end{thm}

We first describe an offline implementation of our algorithm from
Theorem~\ref{thm:weighted-intervals-1} and later describe how to turn it into a
dynamic algorithm.
Before describing our algorithm in full detail, we first sketch its main steps.
Roughly speaking, our offline algorithm called on a cell $Q$ works as
follows.  (1)~We run the $O(1)$-approximation algorithm from
Section~\ref{sec:Weighted-hypercubes} which internally maintains a set of points
$P(Q)$; the points in $P(Q)$ correspond to a superset of the endpoints of
intervals in a $O(1)$-approximate IS. (2)~Based on the points in
$P(Q)$, we define an auxiliary grid (slightly different though than the
auxiliary grid defined in Section~\ref{sec:Weighted-hypercubes}).  After that,
we define \emph{segments} inside $Q$, where each segments corresponds to a
consecutive set of auxiliary grid cells.  (3)~We compute a
$(1+O(\varepsilon))$-approximate solution for each of the previously defined
segments.  To do this, we split the segments into \emph{subsegments} and assume
that we already know solutions for the subsegments. Then we compute a solution
for the segment by running a static algorithm on an IS instance
which is based on the solutions for the subsegments.  (4)~We compute
solutions for each subsegment. To obtain a solution for each subsegment of $Q$,
we show that it suffices to only consider solutions such that either
(a)~a~subsegment contains only $1/\varepsilon$ independent intervals and in this
case we can compute a $(1+\epsilon)$-approximate solution by enumerating all
possible solutions or (b)~the~solution of the subsegment only consists of
segments from the children cells, $Q_1$ and $Q_2$, and in this case we can use
previously computed solutions for the segments of $Q_1$ and $Q_2$. Finally, we
describe how the offline algorithm can be made dynamic; roughly speaking, this
is done by running the previously described procedure only on those grid cells
which are affected by an update operation, i.e., those cells which contain the
interval which was inserted/deleted during the update operation.

We note that our approach does not extend to dimensions $d>1$ because our
structural lemma (Lemma~\ref{lem:structured-solution}) and our algorithm for
computing a $(1+\varepsilon)$-approximate IS $I$ with $|I| \leq
\OoEps$ in time $O_{d,\varepsilon}(\polylog(n))$ (Lemma~\ref{lem:dp-sparse}) do not
extend to higher dimensions.

\textbf{Auxiliary grid and segments.}
Main step~(1) of the algorithm is implemented as follows. We run the algorithm
from Section~\ref{sec:Weighted-hypercubes} which maintains a $O(1)$-approximate
solution; we use the deterministic version of the algorithm since we use it as a
subroutine.  Recall that for each cell $Q$, the $O(1)$-approximation algorithm
maintains a set $P(Q)$ consisting of points in $Q$. Further recall that each
point in $P(Q)$ corresponds to an endpoint of an interval in $\C(Q)$, where the
point has the same weight as its corresponding interval; if a point is the
endpoint of multiple intervals, then its weight is the sum of the weights of the
corresponding intervals. 

We implement main step~(2) as follows.  We define an \emph{auxiliary grid}
$Z(Q)$ inside $Q$. We partition $Q$ into $O( (1/\epsilon)^{3+1/\varepsilon} \log N)$
intervals such that in the interior of each interval the points in $P(Q)$ have a
total weight of at most $\epsilon^{3+1/\varepsilon} w(P(Q))/\log N$ and such
that each interval in $\C'(Q)$ intersects at least two intervals.  The following
lemma shows how the auxiliary grid $Z(Q)$ can be constructed.
\begin{lem}
\label{lem:weighted-intervals-grid}
	Given a cell $Q$ with endpoints $(x_Q,y_Q)$ and a set of points $P(Q)$ stored in the (range
	counting) data structure due to Lemma~\ref{lem:range-counting-data-structure}.
	In time $O((1/\varepsilon)^{3+1/\varepsilon} \log N\lg^2 |P(Q)|)$ 
	we can compute a set of coordinates $Z(Q)=\{z_{1},z_{2},\dots\}\subseteq Q$
	such that
	\begin{itemize}
		\item $x_Q,y_Q \in Z(Q)$,
		\item $k'\cdot\epsilon N/2^{\ell(Q)-1}\in Z(Q)$ for each $k'\in\Z$
			such that $k'\cdot\epsilon N/2^{\ell(Q)-1}\in Q$,
		\item $w(P(Q)\cap(z_{k},z_{k+1}))\le \epsilon^{3+1/\varepsilon} \cdot w(P(Q))/\log N$
			for each $k=1,\dots,|Z(Q)|$, and
		\item $|Z(Q)| = O( (1/\epsilon)^{3+1/\varepsilon}\log N)$.
	\end{itemize}
\end{lem}
\begin{proof}
	The proof is very similar to the proof of
	Lemma~\ref{lem:construct-aux-grid}.  The only difference is that we
	additionally add the points
	$z_{k'} = k'\cdot\epsilon N/2^{\ell(Q)-1}$ for each $k'\in\Z$ such
	that $z_{k'} \in Q$; since $|Q| = N/2^{\ell(Q)-1}$ by definition of the
	hierarchical grid $\G$, this adds $O(1/\varepsilon)$ points to $Z(Q)$.
\end{proof}

We say that an interval $(x,y)$ is \emph{aligned} with a set of points $P$ if
$x,y\in P$; note that $x$ and $y$ might not necessarily be consecutive points in
$P$, i.e., there might be $z\in P$ such that $x < z < y$.
Let $\S(Q)$ be the set of all intervals that are aligned with $Z(Q)$. We refer
to the intervals in $\S(Q)$ as \emph{segments}.  Note that some segments
intersect/overlap and that
$|\S(Q)| = O(|Z(Q)|^2) = O( (1/\epsilon)^{6+2/\varepsilon}\log^2 N)$.

\textbf{Solutions for segments.}
Now we implement main step~(3). We split the segment into subsegments and
provide an algorithm which computes a near-optimal solution for the segment
given the solutions of all subsegments.

Let $Q$ be a cell with children cells $\ch(Q)=\{Q_{1},Q_{2}\}$ and let
$S\in\S(Q)$ be a segment with $S=(s_{1},s_{2})$.  Now we explain how our
algorithm computes a solution $\ALG(Q,S)$. In the following, we assume that for
$r=1,2$, each set $\ALG(Q_r,S')$ has been computed already for each segment $S'
\in \S(Q_r)$, i.e., for both children cells of $Q$ we know ISs for
each of their segments. We will guarantee that each such set $\ALG(Q_r,S')$ is a
$(1+O(\epsilon))$-approximation with respect to the optimal IS
consisting only of intervals in $\C(Q_r)$ contained in $S'$.  We now describe
how we use the solutions $\ALG(Q_r,S')$ to compute $\ALG(Q,S)$ as a
$(1+O(\varepsilon))$-approximate solution for $Q$.

To compute $\ALG(Q,S)$, our algorithm does the following. Initially, it defines
a second auxiliary grid $\Zfine$ based on the auxiliary grids of the two
children cells; note that $\Zfine$ is potentially different from $Z(Q)$.
Formally, we define $\Zfine := Z(Q_{1})\cup Z(Q_{2})\cup\{s_{1},s_{2}\}$, where
$s_1$ and $s_2$ are the endpoints of the current segment $S$. Based on $\Zfine$
we define \emph{subsegments} inside $S$: we define
$\Tfine = \{ (z,z') : z,z' \in \Zfine, z < z', (z,z') \subseteq S \}$ and call
the intervals $T \in \Tfine$ \emph{subsegments}.
Note that since $\Zfine$ contains all endpoints of $Z(Q_1)$ and $Z(Q_2)$, we
can later use the solutions computed for the segments of $Q_1$ and $Q_2$ (which
are aligned with $Z(Q_1)$ and $Z(Q_2)$) in order to obtain solutions for the
subsegments in $\Tfine$.

The next step is to compute a solution $\ALG(Q,T)$ for each subsegment $T \in
\Tfine$. We explain below how the solution $\ALG(Q,T)$ is computed concretely; for
now assume that we have computed it already for all $T\in\Tfine$. 

Using the solutions for the subsegments, we compute a solution for
$\ALG(Q,S)$ that represents the optimal way to combine the solutions for the
subsegments. To do this, we create a new instance of maximum weight IS
of intervals based on $\Tfine$. The new instance
contains one interval $C_{T}:=T$ for each subsegment $T\in\Tfine$
with weight $w_{T}:=w(\ALG(Q,T))$. For this instance,
we construct the corresponding interval graph with $n'=|\Tfine|$ vertices
and $m'=O(|\Tfine|^{2})$ edges and then we apply the $O(n'+m')=O(|\Tfine|^{2})$
time algorithm in~\cite{frank75some} which computes an optimal solution
for this instance. Let $\T^{*}\subseteq\Tfine$ denote the subsegments
corresponding to this optimal solution. We define
$\ALG(Q,S):=\bigcup_{T\in\T^{*}}\ALG(Q,T)$.

\textbf{Solutions for subsegments.}
Next, we implement main step~(4) which computes a solution for a subsegment.
This solution will either consist of the solutions of segments of children cells
or of ``sparse'' solutions which contain at most $1/\varepsilon$ intervals.

Let $Q$ be a cell with children $\ch(Q) = \{Q_1,Q_2\}$.  Now we explain how to
compute a solution $\ALG(Q,T)$ for a subsegment $T\in\Tfine$. We compute
(explained next) a dense solution $\ALG_\dense(Q,T)$ and a sparse solution
$\ALG_\sparse(Q,T)$. After this, the algorithm sets $\ALG(Q,T)$ to the solution
with higher weight among $\ALG_\dense(Q,T)$ and $\ALG_\sparse(Q,T)$.

To compute the dense solution $\ALG_\dense(Q,T)$, we check if
$T\subseteq Q_r$ for some $r=1,2$, i.e., if $T$ is completely contained in one
of the children cells. If this is the case, then we set
$\ALG_\dense(Q,T) = \ALG(Q_r,T)$. Note that this is possible since we have
access to all solutions $\ALG(Q_r,T')$ for the children cells and since $T\in
\S(Q_r)$ by definition of $\Tfine$. If the above check fails, e.g., because
$\ch(Q) = \emptyset$ or because $T$ overlaps with both $Q_1$
and $Q_2$, we set $\ALG_\dense(Q_r,T) = \emptyset$.

To compute the sparse solution $\ALG_\sparse(Q,T)$, we find a
$(1+O(\varepsilon))$-approximate solution of the maximum weight IS
of intervals within $T$ which contains at most $1/\varepsilon$ intervals.  We do
this by arguing that for each interval we only $O(1/\varepsilon)$ need to
consider different weight classes and, hence, we only need to consider
$(1/\varepsilon)^{O(1/\varepsilon)}$ different sequences of interval weight
classes. Then we exhaustively enumerate all of these sequences and check for
each sequence if such a sequence of intervals exists in $T$.
Lemma~\ref{lem:dp-sparse} explains how this is done in detail and we prove it in
Section~\ref{sec:dp-sparse}.
\begin{lem}
\label{lem:dp-sparse}
	Consider a subsegment $T\in\Tfine$. In time $(\OoEps)^{O(1/\varepsilon)}\lg^{2}n$ we can
	compute a set $\ALG_{\sparse}(Q,T)$ which contains at most $1/\epsilon$
	intervals from $\C(Q)$ that are all contained in $T$.  Moreover, if
	$\overline{\OPT}(Q,T)$ is the maximum weight IS with this
	property, then
	$w(\ALG_\sparse(Q,T))\ge(1+\epsilon)^{-1}w(\overline{\OPT}(Q,T))$.
\end{lem}

To store the solutions $\ALG(Q,T)$, we store all intervals
explicitly if it holds that $\ALG(Q,S) = \ALG_\sparse(Q,T)$ and if $\ALG(Q,S) =
\ALG_\dense(Q,T)$, we store a pointer to the respective set $\ALG(Q_{j},T)$.

\textbf{The global solution for $Q^*$.} Finally, the global solution $\SOL$ which we
output is the solution for the cell $Q^*$ stored in $\ALG(Q^{*},[0,N])$, where
$Q^*$ is the cell at level $0$ containing the whole space $[0,N]$. Note that the
above algorithm indeed computed such a solution because $Q^*$ has endpoints $0$
and $N$ and, hence, $[0,N] \in Z(Q^*)$. In order to return $\SOL$, we first
output all intervals that are stored in $\ALG(Q^{*},[0,N])$ explicitly.  After
that, we recursively output all solutions in sets $\ALG(Q',S')$ to which we
stored a pointer in $\ALG(Q^{*},[0,N])$.

\textbf{Dynamic algorithm.}
To implement the above offline algorithm as a dynamic algorithm, we use the
dynamic hierarchical grid decomposition due to Section~\ref{sec:framework}
together with its adjustment described in Section~\ref{sec:adjustment}. Now we
describe how to implement the preprocessing, update and query operations.

\emph{Preprocessing.} The preprocessing of the algorithm is exactly the same as
in Section~\ref{sec:adjustment}.

\emph{Update.} Suppose that in an update, an interval $C_{i}$ is inserted or
deleted.  First, we run the algorithm from Section~\ref{sec:Weighted-hypercubes}
for this update. This updates the set $P(Q)$ for each cell $Q$ such that
$C_{i}\in\C(Q)$ as described in Section~\ref{sec:adjustment}. Then, for each
cell $Q$ for which $P(Q)$ changes, we recompute the sets $Z(Q)$ and $\S(Q)$ as
described in the offline algorithm.  We sort the latter cells decreasing by
level and for each such cell $Q$, we recompute $\ALG(Q,S)$ for each $S\in\S(Q)$
as described for the offline algorithm.

\emph{Query.} We return the IS $\ALG(Q^*, [0,N])$ as described
above for the offline algorithm.

\subsection{Analysis}
We now proceed to the analysis of the previously presented dynamic algorithm.
Fix any cell $Q$ for the rest of this subsection and let $\OPT(Q)$ denote the
maximum weight IS for $Q$ which only consists of intervals from
$\C(Q)$.

Our proof proceeds as follows. We start by performing some technical
manipulations of $\OPT(Q)$. Next, we prove that for each segment there exists a
near-optimal structured solution (Lemma~\ref{lem:structured-solution}).  We use
this fact to show that our algorithm computes a near-optimal solution for each
segment of $Q$ (Lemma~\ref{lem:induction-cells}), which implies that our global
solution is $(1+\varepsilon)$-approximate
(Lemma~\ref{lem:weighted-intervals:almost-optimal}). We conclude by analyzing
the update time of the algorithm
(Lemma~\ref{lem:weighted-intervals:update-time}).

We start with several simplifications of $\OPT(Q)$. For each
$C_{i}\in\OPT(Q)\cap\C'(Q)$, we say that $C_{i}$ is \emph{light} if $w_{i}\le
(\epsilon^{2}/\log N) \cdot w(P(Q))$ and \emph{heavy} otherwise. We
delete all light intervals from $\OPT(Q)$. Let $\OPT'(Q)\subseteq\OPT(Q)$ denote
the remaining (heavy) intervals from $\OPT(Q)$.

Next, we show that $\OPT'(Q)$ has large weight and that for each subinterval
$I\subseteq Q$, the weight of the points in $I\cap P(Q)$ yields an upper bound
for the profit that $\OPT'(Q)$ obtains in $I$.
\begin{lem}
\label{lem:estimate-weight-OPT}
	We have that $w(\OPT'(Q))\ge(1-O(\epsilon))w(\OPT(Q))$ and there exists a
	constant $\kappa = \Theta(1)$ such that for each interval
	$I\subseteq Q$ the total weight of the intervals in
	$\{C_{i}\in\OPT'(Q)|C_{i}\subseteq I\}$ is bounded by $\kappa w(P(Q)\cap I)$.
\end{lem}
\begin{proof}
	This follows from the argumentation in the proof of
	Lemma~\ref{lem:apx-ratio-weighted-hypercubes}.
\end{proof}

Next, we group the intervals such that the weight of intervals in
different groups differs by at least a $1/\varepsilon$ factor while losing a
factor of at most $1+\epsilon$ in the resulting solution.  To do so, we employ a
shifting step.  Given an offset $a'\in\{0,\dots,1/\epsilon-1\}$, we delete the
intervals in $\OPT'(Q,a')\subseteq\OPT'(Q)$, where we define $\OPT'(Q,a')$ to be
all intervals $C_{i}\in\OPT'$ for which
$w_{i}\in\bigcup_{k\in\N}[\frac{1}{\epsilon}^{a'+k/\epsilon},\frac{1}{\epsilon}^{a'+k/\epsilon+1})$.
Then we group the intervals into supergroups according to ranges of weights of
the form
$[\frac{1}{\epsilon}^{a'+k/\epsilon+1},\frac{1}{\epsilon}^{a'+(k+1)/\epsilon})$
with $k\in\N$ such that if two intervals in $\OPT'(Q)\setminus\OPT'(Q,a')$ are
in different supergroups then their weights differ by a factor of at least
$1/\epsilon$.
\begin{lem}
	\label{lem:shifting}
	There is a value for $a'\in\{0,\dots,1/\epsilon-1\}$
	such that $w(\OPT'(Q,a'))\le\epsilon\cdot\OPT'(Q)$.
\end{lem}
\begin{proof}
	Consider $a', a'' \in \{0,\dots,1/\epsilon-1\}$ with $a' \neq a''$. Now observe
	that in this case the sets
	$\bigcup_{k\in\N}[\frac{1}{\epsilon}^{a'+k/\epsilon},\frac{1}{\epsilon}^{a'+k/\epsilon+1})$
	and
	$\bigcup_{k\in\N}[\frac{1}{\epsilon}^{a''+k/\epsilon},\frac{1}{\epsilon}^{a''+k/\epsilon+1})$
	are disjoint.
	Since the weight $w_i$ of each $C_i \in \OPT'$ belongs to exactly one such set and
	there are $1/\varepsilon$ different values for $a'$, there must exist a
	shift $a'$ with the property claimed in the lemma.
\end{proof}

For the value $a'$ due to Lemma~\ref{lem:shifting} we define
$\OPT^*(Q):=\OPT'(Q)\setminus\OPT'(Q,a')$.  For each segment $S\in\S(Q)$, let
$\OPT^*(Q,S)$ denote all intervals from $\OPT^*(Q)$ that are contained in $S$. 

\textbf{A structured solution.} 
Now our goal is to prove that for each segment $S\in\S(Q)$, the set
$\ALG(Q,S)\subseteq\C(Q)$ computed by our algorithm indeed approximates
$\OPT^*(Q,S)$ almost optimally. To this end, fix any segment $S\in\S(Q)$.  We
argue that there is a structured near-optimal solution $\OPT''(Q,S)$ for $S$
such that $S$ is divided into subsegments $\T$ that are aligned with $\Zfine$
and where the solution of each subsegment $T\in\T$ either contains only
$1/\epsilon$ intervals in total or only intervals from $\C(Q_{1})$ or from
$\C(Q_{2})$. We formalize this notion of a structured solution in the following
definition. 

\begin{definition}
\label{def:structured-solution}
	A solution $\OPT''(Q,S)\subseteq\OPT^*(Q,S)$ is \emph{structured} if there
	exists a set of disjoint subsegments $\T=\{T_{1},\dots,T_{|\T|}\}$ aligned
	with $\Zfine$ such that each interval $C_{i}\in\OPT''(Q,S)$ is contained in
	some subsegment $T_{j}\in\T$ and for each segment $T_{j}\in\T$ it holds that
	$T_{j}\subseteq S$ and either
	\begin{enumerate}
		\item $T_j$ contains at most $1/\epsilon$ intervals from $\OPT''(Q,S)$, or
		\item there exists $r\in\{1,2\}$ such that $T_{j}\subseteq Q_{r}$ and if an interval
			$C_{i}\in\OPT''(Q,S)$ is contained in $T_{j}$ then $C_{i}\in\C(Q_{r})$.
	\end{enumerate}
\end{definition}
Now consider a structured solution $\OPT''(Q,S)$ and a set of disjoint
subsegments $\T$ as per Definition~\ref{def:structured-solution}.
For $T_j \in \T$, we define $\OPT_{T_{j}}''(Q,S) := \OPT''(Q,S) \cap T_{j}$.
Furthermore, we let $\T_{\sparse}\subseteq\T$ and $\T_{\dense}\subseteq\T$
denote the sets of subsegments in $\T$ for which the first and the second case
of the definition applies, respectively. Notice that since $\T$ is aligned with
$\Zfine$, we have that $\T \subseteq \Tfine$.
	
The following lemma shows that there exists a structured solution which is
within a $(1 + O(\varepsilon))$-factor of $\OPT^*(Q,S)$; we will prove the lemma
in Section~\ref{sec:structured-solution}.
\begin{lem}
\label{lem:structured-solution}
	There exists a structured solution $\OPT''(Q,S)\subseteq\OPT^*(Q,S)$
	and a set of disjoint subsegments $\T=\{T_{1},\dots,T_{|\T|}\}$
	such that for a fixed $\kappa' = \Theta(1)$ that is independent of $Q$ and
	$S$,
	\begin{align*}
		(1+\kappa'\cdot \epsilon) \sum_{T_{j}\in\T_{\sparse}}w(\OPT_{T_{j}}''(Q,S))+\sum_{T_{j}\in\T_{\dense}}w(\OPT_{T_{j}}''(Q,S))\ge w(\OPT^*(Q,S)).
	\end{align*}
\end{lem}

Given Lemma~\ref{lem:structured-solution}, observe $\T\subseteq\Tfine$
since $\T$ is aligned with $\Zfine$ and since $\Tfine$ contains all intervals
which are aligned with $\Zfine$. Further observe that for a subsegment
$T\in\T_\sparse$, the solution $\ALG_\sparse(Q,T)$ provides a
$(1+\varepsilon)$-approximate solution for $T$. Also, if $T\in\T_\dense$, then we will argue by induction
that $\ALG_\dense(Q,T) = \ALG(Q_r,T)$ is a
$(1+\varepsilon)$-approximate solution for $\OPT_{T_{j}}''(Q,S)$.

We conclude the section with thre lemmas which after combining with
Lemma~\ref{lem:framework-dynamic} complete the proof of
Theorem~\ref{thm:weighted-intervals-1}.  The first two lemmas establish that
$\ALG(Q,S)$ is a $(1+O(\varepsilon))$-approximation of $\OPT^*(Q,S)$ and the
last lemma bounds the update time of the algorithm.

\begin{lem}
\label{lem:induction-cells}
	Let $Q' \subseteq Q$ be a cell and let $S'\in\S(Q')$ be a segment. Then
	\begin{align*}
		w(\ALG(Q',S'))\ge (1+\varepsilon)^{-1} (1+\kappa' \varepsilon)^{-1} w(\OPT^*(Q',S')).
	\end{align*}
\end{lem}
\begin{proof}
	We prove by induction over the level of $Q'$ in decreasing levels. In the
	base case, we obtain the claim based on Lemma~\ref{lem:dp-sparse}. In the
	induction step, we use Lemma~\ref{lem:structured-solution} to obtain the
	desired inequality.

	First, suppose that $Q'$ has level $\ell(Q') = \lg N$. Let $S'\in\S(Q')$ and
	let $\T$ be as in Definition~\ref{def:structured-solution}.  First, since
	$\T$ is aligned with $\Zfine$, we have that for each $T_j \in \T$ it holds
	that $T_j\in\Tfine$. Thus, the static algorithm which we run to compute
	$\ALG(Q',S')$ could pick the solutions for the intervals $T_j\in\T$ and thus
	we get $w(\ALG(Q',S')) \geq \sum_{T_j\in\T} w(\ALG(Q',T_j))$.  Second,
	observe that since $Q'$ has level $\lg N$, we must have that
	$\T_{\dense} = \emptyset$ and $\T = \T_{\sparse}$.
	Observe that for each subsegment $T_j\in\T_{\sparse}$ the solution
	$\ALG_\sparse(Q',T_j)$ is a $(1+\varepsilon)$-approximation of
	$\OPT''_{T_j}(Q',S)$ by Lemma~\ref{lem:dp-sparse}.  This gives
	$\sum_{T_j\in\T} w(\ALG(Q',T_j)) \geq \sum_{T_j\in\T}
	(1+\varepsilon)^{-1}w(\OPT_{T_j}''(Q',S'))$.
	Now by using the inequality from Lemma~\ref{lem:structured-solution} with
	$\T_{\dense} = \emptyset$ and $\T = \T_{\sparse}$, we obtain that
	$\sum_{T_j\in\T} (1+\varepsilon)^{-1}w(\OPT_{T_j}''(Q',S'))
		\geq (1+\varepsilon)^{-1} (1+ \kappa' \cdot \varepsilon)^{-1}
		w(\OPT^*(Q',S'))$.
	Summarizing all of the previous steps, we have that
	\begin{align*}
		w(\ALG(Q',S'))
		&\geq \sum_{T_j\in\T} w(\ALG(Q',T_j)) \\
		&\geq \sum_{T_j\in\T} (1+\varepsilon)^{-1}w(\OPT_{T_j}''(Q',S')) \\
		& \geq (1+\varepsilon)^{-1} (1+ \kappa' \cdot \varepsilon)^{-1} w(\OPT^*(Q',S')).
	\end{align*}

	Now suppose that $Q'$ has level $\ell(Q') < \lg N$ and has children
	$Q'_1,Q'_2$. We can assume by induction hypothesis that for $r=1,2$, we have
	that $w(\ALG(Q'_r,S_r')) \geq
	(1+\epsilon)^{-1}(1+\kappa'\cdot\epsilon)^{-1}w(\OPT^*(Q'_r,S_r'))$
	for each $S_r' \in \S(Q'_r)$. 
	We again use Lemma~\ref{lem:structured-solution} to obtain the sets
	$\T$, $\T_\dense$ and $\T_\sparse$. Now for each $T_j \in \T_\sparse$ we
	have that $\ALG(Q',T_j)\ge (1+\epsilon)^{-1}w(\OPT''_{T_j}(Q',S'))$ by
	Lemma~\ref{lem:dp-sparse}. For each $r\in \{1,2\}$ and each
	$T_j\in\T_\dense$ with $T_j \subseteq Q'_r$ we have that 
	$\ALG(Q'_r,T_j)\ge (1+\epsilon)^{-1}(1+\kappa'\cdot \varepsilon)^{-1}
	\OPT^*(Q'_r,T_j)$ by induction hypothesis. We conclude that 
	\begin{align*}
		w(\ALG(Q',S'))
		\geq& \sum_{T_j\in\T_\sparse} w(\ALG(Q',T_j)) + \sum_{T_j\in\T_\dense} w(\ALG(Q',T_j)) \\
		\geq& \sum_{T_j\in\T_\sparse} (1+\epsilon)^{-1}w(\OPT''_{T_j}(Q',S')) \\
		&+\sum_{r=1}^2\sum_{T_j\in\T_\dense:T_j\in Q'_r}
		(1+\epsilon)^{-1}(1+\kappa'\cdot \varepsilon)^{-1} w(\OPT^*(Q'_r,T_j)) \\
		\geq& \sum_{T_j\in\T_\sparse} (1+\epsilon)^{-1}w(\OPT''_{T_j}(Q',S')) \\
		&+ \sum_{T_j\in\T_\dense} (1+\epsilon)^{-1}(1+\kappa'\cdot \varepsilon)^{-1} w(\OPT''_{T_j}(Q',S')) \\
		=& (1+\epsilon)^{-1}(1+\kappa'\cdot \varepsilon)^{-1}
			\left[ \sum_{T_j\in\T_\sparse} (1+\kappa' \cdot \epsilon) w(\OPT''_{T_j}(Q',S'))\right.  \\
			&\hspace{4cm} + \left.\sum_{T_j\in\T_\dense}  w(\OPT''_{T_j}(Q',S')) \right] \\
		\geq& (1+\epsilon)^{-1}(1+\kappa'\cdot \varepsilon)^{-1}w(\OPT^*(Q',S')).
	\end{align*}
	In the above computation, the first inequality holds since
	$\T\subseteq\Tfine$ (as $\T$ is aligned with $\Zfine$) and the exact
	algorithm can pick, for example, all intervals in subsegments of $\T$ as its
	solution. 
	The second inequality follows from Lemma~\ref{lem:dp-sparse} for the
	subsegments in $\T_{\sparse}$ and from the induction hypothesis for the
	subsegments in $\T_{\dense}$.
	The third inequality holds since by definition of $\T$ (see
	Definition~\ref{def:structured-solution}), we have that
	$\OPT''_{T_j}(Q',S') \subseteq \OPT^*(Q_r',T_j)$ and that
	for all $C_i \in \OPT''_{T_j}(Q',S')$ it holds that
	$C_i \in \C(Q'_r)$; this implies that indeed
	$w(\OPT''_{T_j}(Q',S')) \leq (\OPT^*(Q_r',T_j))$.
	The equality is simple algebra.
	The last inequality follows from Lemma~\ref{lem:structured-solution}.
\end{proof}

\begin{lem}
\label{lem:weighted-intervals:almost-optimal}
	We have that $w(\SOL) \geq (1-O(\varepsilon))w(\OPT)$.
\end{lem}
\begin{proof}
	The claim follows since
	$\SOL = \ALG(Q^*,[0,N])$ and Lemma~\ref{lem:induction-cells} applied with $Q'=Q^*$ and
	$S'=[0,N]$ yields that
	$w(\SOL) = w(\ALG(Q^*,[0,N])) \geq (1-O(\epsilon))w(\OPT^*(Q^*,[0,N])) \geq (1-O(\epsilon))w(\OPT(Q^*,[0,N]))$.
\end{proof}

\begin{lem}
\label{lem:weighted-intervals:update-time}
	The algorithm has worst-case update time 
	$\updateweightedintervals$.
\end{lem}
\begin{proof}
	First, we run the update step of the algorithm from
	Section~\ref{sec:Weighted-hypercubes} which has worst-case update time
	$(1/\varepsilon)^{O(1/\varepsilon)} \lg W \lg n \lg^3 N$. Then we
	need to update $\log N$ grid cells.  For each such cell we must construct
	the grid $Z(Q)$ which takes
	time $O((1/\varepsilon)^{1/\varepsilon} \log N \lg^2 |P(Q)|)$.
	Note that for each cell $Q$ the number of segments is bounded by
	$|Z(Q)|^2\leq (1/\varepsilon)^{O(1/\varepsilon)} \lg^2 N$.
   	Furthermore, for each segment we have to consider 
   	$|\Tfine|\le|\Zfine|^{2}\le (1/\epsilon)^{O(1/\varepsilon)}\log^2 N$
	subsegments. Then for each subsegments, the algorithm spends time
	$(\OoEps)^{O(1/\varepsilon)}\lg^{2}n$ for the computation due to
	Lemma~\ref{lem:dp-sparse}. This takes total time
	$(1/\varepsilon)^{O(1/\varepsilon)} \lg^4 N \lg^{2}n$.
	Furthermore, for computing the solution $\ALG(Q,S)$ we need to solve the
	static maximum independent set problem on an interval graph of
	$O(|\Tfine|^2)$ edges which takes $O(|\Tfine|^2)$ time.

	For all $\log N$ cells together that we need to update 
	we can upper bound all of these running times 
	by $\updateweightedintervals$. 
\end{proof}

\subsection{Proof of Lemma~\ref{lem:dp-sparse}}
\label{sec:dp-sparse}
Consider $\overline{\OPT}(Q,T)$. By Lemma~\ref{lem:estimate-weight-OPT}
we have that $w(\overline{\OPT}(Q,T)) \leq w(\OPT(Q,T)) \leq \kappa w(P(Q) \cap T)$,
where $\kappa$ is the constant in Lemma~\ref{lem:estimate-weight-OPT}.
Hence, we can remove all intervals with weight less than $\varepsilon^2 w(P(Q) \cap T)$ 
from $\overline{\OPT}(Q,T)$ and the resulting independent set still has weight
at least $(1-\varepsilon)w(\overline{\OPT}(Q,T))$. Call the resulting
independent set $\overline{\OPT}'(Q,T)$.

Now set $k_u$ to the smallest integer such that $(1+\varepsilon)^{k_u} \geq
\kappa w(P(Q) \cap T)$ and let $k_l$ be the largest integer such that
$(1+\varepsilon)^{k_l} \leq \varepsilon^2 w(P(Q) \cap T)$; recall that
$\kappa=\Theta(1)$ as in Lemma~\ref{lem:estimate-weight-OPT} is a constant not
depending on $n$ or $\varepsilon$.
Now note that $k_u - k_l = O(\lg_{1+\varepsilon} \varepsilon^2) = O((\lg \varepsilon)/\varepsilon) = O( (\OoEps)^{2})$.

Now we compute an independent set $\ALG_\sparse(Q,T)$ by exhaustively
enumerating all possible sequences
$\K=(\hat{K},\hat{k_{1}},\dots,\hat{k}_{\hat{K}})$ with the following properties: 
\begin{itemize}
	\item $\hat{K}\in\{0,1,\dots,\OoEps\}$ is a guess for $|\overline{\OPT}'(Q,T)|$. Note that
		there are $O(1/\varepsilon)$ choices for $\hat{K}$. 
	\item Suppose that $\overline{\OPT}'(Q,T)$ has indeed $\hat{K}$ elements and has the form
		$\overline{\OPT}'(Q,T)=\{C_{1},\dots,C_{\hat{K}}\}$, where $C_{j}$ lies completely
		on the left of $C_{j+1}$ for each $j$. Then $\hat{k_{i}}\in\N$
		is a guess for the weight of $C_{i}$ in the sense that $C_{i}$ has
		weight $(1+\varepsilon)^{\hat{k_{i}}}\leq w_{i}<(1+\varepsilon)^{\hat{k_{i}}+1}$.
		Since we removed all intervals with weight less than $\varepsilon^2 w(P(Q) \cap T)$
		from $\overline{\OPT}'(Q,T)$, all $C_{i}\in\overline{\OPT}'(Q,T)$
		have weights such that $(1+\epsilon)^{k_l}\le w_{i}\le(1+\epsilon)^{k_u}$.
		Hence, we only need to consider $O( (\OoEps)^{2})$ values for $\hat{k_{i}}$ and
		thus we have $O((\OoEps)^{\hat{K}+2})=(\OoEps)^{O(1/\varepsilon)}$
		choices for picking $(\hat{k_{1}},\dots,\hat{k}_{\hat{K}})$. 
\end{itemize}
Now given $T=(t_1,t_2)$ and a sequence
$\K=(\hat{K},\hat{k_{1}},\dots,\hat{k}_{\hat{K}})$ as above, we compute a
candidate solution $\ALG_\sparse^{\K}(Q,T)$
as follows. First, we initialize $\ALG_\sparse^{\K}(Q,T)=\emptyset$ and $t=t_1$.
Then for $i=1,\dots,\hat{K}$ in increasing order, we use the
fourth property from Lemma~\ref{lem:data-structure} on $\D_{\hat{k_{i}}}$
to find the interval $C_{i}=[x,y)\in\C_{\hat{k_{i}}}(Q)$ in the range
$[t,t_2)$ with smallest $y$-coordinate (recall that $\D_{\hat{k_{i}}}$ maintains
all intervals with weights in the range
$[(1+\varepsilon)^{\hat{k_i}}, (1+\varepsilon)^{\hat{k_i}+1})$).
If such a $C_{i}$ does not
exist, then the subroutine aborts and proceeds with the next sequence
$\K$, otherwise, $C_{i}$ is added to $\ALG_\sparse^{\K}(Q,T)$ and $t$ is
set to $y$.

Now the subroutine computes $\ALG_\sparse^{\K}(Q,T)$ as approximation of
$\overline{\OPT}'(Q,T)$
by iterating over all sequences $\K$ as defined above and running
the above routine for computing $\ALG_\sparse^{\K}(Q,T)$. Then the subroutine
sets $\ALG_\sparse^{\K}(Q,T)$ to the solution $\ALG_\sparse^{\K}(Q,T)$ with the
maximum weight.

To bound the running time of the procedure, note that there are
$(\OoEps)^{O(1/\varepsilon)}$ choices for $\K$. For each $\K$, computing
$\ALG_\sparse^{\K}(Q,T)$ takes time
$O(\OoEps\lg^{2}n)$ because each call to $D_{\hat{k_i}}$ takes time $O(\lg^2 n)$
by Lemma~\ref{lem:data-structure}.  Thus, the computing $\ALG_\sparse(Q,T)$ takes time
$(\OoEps)^{O(1/\varepsilon)}\lg^{2}n$.

Next, we show that 
$w(\ALG_\sparse(Q,T))\geq(1-\varepsilon)\overline{\OPT}'(Q,T)$. 
To see this, suppose that $\overline{\OPT}'(Q,T)=\{C_{1},\dots,C_{\hat{K}}\}$ for
suitable $\hat{K}\leq\OoEps$, where $C_{j}$ is to the left of $C_{j+1}$
for all $j$. Furthermore, let $\hat{k_{j}}$
be such that $(1+\varepsilon)^{\hat{k_{j}}}\leq
w_{j}<(1+\varepsilon)^{\hat{k_{j}}+1}$ for all $j=1,\dots,\hat{K}$.
Then the definition of the subroutine
implies that the subroutine has computed a solution $\ALG_\sparse^{\K}(Q,T)$
for the sequence $\K=(\hat{K},\hat{k_{1}},\dots,\hat{k_{K}})$.
Hence, $w(\ALG_\sparse(Q,T))\geq w(\ALG_\sparse^{\K}(Q,T))$.

Now consider the solution $\ALG_\sparse^{\K}(Q,T)=\{\hat{C_{1}},\dots,\hat{C}_{\hat K}\}$
computed by the algorithm for sequence $\K$ as above, where $\hat{C_{j}}$
is on the left of $\hat{C_{j+1}}$ for all $j$. First, note that
the subroutine did not abort since for $\hat{C_{j}}=(\hat{x_{j}},\hat{y_{j}})$
and $C_{j}=(x_{j},y_{j})$ we must have that $\hat{y_{j}}\leq y_{j}$
(by definition of the subroutine and the fourth property of Lemma~\ref{lem:data-structure}).
Second, observe that the weights of $\hat{C_{j}}$ and $C_{j}$ differ
by at most a $(1+\varepsilon)$-factor since $\hat{C_{j}}\in\C_{\hat{k_{i}}}(Q)$
and $C_{j}\in\C_{\hat{k_{i}}}(Q)$ by definition of the subroutine.
This implies that
$w(\ALG_\sparse^{\K}(Q,T))\geq(1-\varepsilon)w(\overline{\OPT}'(Q,T))\geq(1-O(\varepsilon))w(\overline{\OPT}(Q,T))$.

\subsection{Proof of Lemma~\ref{lem:structured-solution}}
\label{sec:structured-solution}
The proof has five main steps. (1) We start by defining important intervals which are
intervals from $\OPT^*(Q,S)$ which have high weight and which we will be
charging throughout the proof (we will never charge non-important intervals).
(2) We prove a technical claim which states that every important interval
must overlap with a point from $\Zfine$. (3) We show how to construct
$\OPT''(Q,S)$ as $(1+O(\varepsilon))$-approximation of $\OPT^*(Q,S)$ in two
substeps: (3.1) We show that for each important interval $C_i$ which partially
overlaps with a grid slice $(z_k,z_{k+1})$, we can delete all non-important
intervals from $\OPT^*(Q,S)$ which overlap with $(z_k,z_{k+1})$.  (3.2) We use
a shifting argument and remove every $1/\varepsilon$'th important interval from
$\OPT''(Q,S)$; this will be useful to satisfy the first property of
Definition~\ref{def:structured-solution} later.  (4) Using the solution $\OPT''(Q,S)$
we show how the set $\T$ can be constructed; this shows that $\OPT''(Q,S)$ and
$\T$ together yield a structured solution.  (5) Finally, we prove that the
inequality which we claim in the lemma holds; this implies that our structured
solution is almost optimal.

(1): Consider $\OPT^*(Q,S)$. Since we deleted all light intervals from $\OPT^*(Q,S)$ we
have for each $C_{i}\in\OPT^*(Q,S)\cap\C'(Q)$ that
$w_{i}\ge (\epsilon^{2}/\log N)\cdot w(P(Q))\ge(\epsilon^{2}/\log N) \cdot \max\{w(P(Q_{1}),w(P(Q_{2})\}$.
Let $k$ be the largest integer such that
$\tilde{w}:=\frac{1}{\epsilon}^{a'+k/\epsilon+1}\le (\epsilon^{2}/\log N) \cdot w(P(Q))$.
Observe that $\tilde{w} \ge w(P(Q))\epsilon^{2+1/\epsilon}/\log N$ since
otherwise we could have picked a larger value for $k$ in the definition of
$\tilde{w}$.
Note that due to the shifting step, if $w_{i}<\tilde{w}$ for an interval
$C_{i}$ then $w_{i}<\epsilon\tilde{w}$. We say that an interval
$C_{i}\in\OPT^*(Q,S)$ is \emph{important} if $w_{i}\ge\tilde{w}$. 
Note that important intervals can be intervals in $\C'(Q)$ assigned to $Q$ and
also intervals in $\C(Q) \setminus \C'(Q)$ assigned to the children cells of $Q$.

(2) We prove the following claim: Each important interval $C_{i}$ overlaps at
least with one point in $\Zfine$. Indeed, if $C_{i}\in\C'(Q)$ then $C_i$ has size
$|C_i| \geq \varepsilon N / 2^{\ell(Q)-1}$ (by definition of $\C'(Q)$) and thus
the claim follows since (by definition of $Z(Q_r)$) we have that $C_i$ overlaps
one of the points $k'\cdot\epsilon N/2^{\ell(Q_{r})-1}\in Z(Q_{r})$ for
$r\in\{1,2\}$ and for each $k'\in\mathbb{N}$ such that $k'\cdot\epsilon
N/2^{\ell(Q_{r})-1}\in Q_{r}$.  If $C_{i}\notin\C'(Q)$ then observe that $w_i \geq \tilde{w}$
(since $C_i$ is important) and that the auxiliary grid satisfies
$w(P(Q_{r})\cap(z_{k},z_{k+1}))\le w(P(Q_{r}))\epsilon^{3+1/\varepsilon}/\log
N\le w(P(Q))\epsilon^{3+1/\varepsilon}/\log N \leq \varepsilon\tilde{w}$ 
for any two consecutive points $z_{k},z_{k+1}\in Z(Q_{r})$ for $r\in\{1,2\}$.
This implies that $C_i$ must overlap with one of the points of the auxiliary
grid.

(3.1): Now we change $\OPT^*(Q,S)$ as follows. Whenever an important interval $C_{i}$
partially overlaps a grid slice $(z_{k},z_{k+1})$ then we delete all intervals
from $\OPT^*(Q,S)$ that are not important and that overlap a point in
$(z_{k},z_{k+1})$; we do \emph{not} do this if $(z_k,z_{k+1})$ is completely
contained in $C_i$, i.e., if $(z_k,z_{k+1}) \subseteq C_i$.  Now we show that
the total weight of deleted non-important intervals is bounded by
$(2+2\kappa)\epsilon\tilde{w}\le(2+2\kappa)\epsilon w_{i}$, 
where $\kappa$ is as in Lemma~\ref{lem:estimate-weight-OPT}: Consider a
slice $(z_k,z_{k+1})$ with which $C_i$ partially overlaps. First, note that at
most two non-important intervals can overlap with the endpoints $z_k$ and
$z_{k+1}$; deleting these intervals costs at most $2 \varepsilon \tilde{w}$.
Second, we delete the intervals $C^* \in \OPT^*(Q,S)$ which are completely
contained in $(z_k,z_{k+1})$, i.e., those $C^* \in \OPT^*(Q,S)$ such that
$C^* \subseteq (z_k,z_{k+1})$. Deleting these intervals costs at most
$w(\OPT^*(Q,(z_k,z_{k+1})))$.  By definition of $\Zfine$ and
Lemma~\ref{lem:estimate-weight-OPT}, we have that
$w(\OPT^*(Q,(z_k,z_{k+1}))) \leq \kappa w(P(Q) \cap (z_k,z_{k+1}))
		\leq \kappa \cdot \epsilon^{3+1/\varepsilon} w(P(Q))/\log N
		\leq \kappa \cdot \varepsilon \tilde{w}$. 
We conclude that the previous two steps incur costs of at most
$(1+\varepsilon)\kappa \tilde{w}$ for the slice $(z_k,z_{k+1})$. 
Since $C_i$ can partially overlap with at most two slices $(z_k,z_{k+1})$ and
$(z_{k'},z_{k'+1})$, we can charge their total cost $(2+2\kappa)\tilde{w}$ to
the weight $w_i \geq \tilde{w}$ of $C_i$.  This proves that the total of weight
of all deleted non-important intervals is at most $(2+2\kappa)\varepsilon w_i$.

(3.2): Let $C_{i_{1}},\dots,C_{i_{K}}$ denote the important intervals in
$\OPT^*(Q,S)$.  We next show how to choose a suitable offset
$a''\in\{0,\dots,1/\epsilon-1\}$ such that deleting each important interval
$C_{i_{k}}$ with $k\equiv a''\bmod1/\epsilon$ does not decrease the weight of
$\OPT^*(Q,S)$ too much.  As there are $1/\varepsilon$ such offsets and for each
offset the sets of deleted intervals are disjoint, there is choice for $a''$
such the deleted intervals have a total weight of at most
$\epsilon\cdot\sum_{k=1}^{K}w_{i_{k}}$.  Let $\OPT''(Q,S)$ denote the remaining
important intervals from $\OPT^*(Q,S)$ and observe that the total weight of all
deleted important intervals from $\OPT^*(Q,S)$ is $O(\varepsilon w(\OPT^*(Q,S))$;
we charge this cost to the remaining important intervals in $\OPT''(Q,S)$
proportionally to their weights.  Thus, $w(\OPT''(Q,S)) \geq
(1-O(\varepsilon))w(\OPT^*(Q,S))$.

(4): Next, we define the subsegments $\T$. For each $r\in\{1,2\}$, $\T$ consists of
one subsegment for each interval of the form $(z,z')$ with $z,z'\in \Zfine$ such
that (a) $(z,z')\subseteq Q_{r}$, (b) no important interval
intersects with $(z,z')$ and (c) among all pairs $(z,z')$ with properties (a)
and (b), $z' - z$ is maximal; these are the subsegments $\T_{\dense}$ mentioned
in the second point of Definition~\ref{def:structured-solution}. Note that for
an interval $(z,z')$ as just defined, we indeed have that for each $C_i \in
\OPT''(Q,S)$ with $C_i \subseteq (z,z')$ it holds that $C_i \in \C(Q_r)$: Since
$C_i$ is non-important we have that $w_i \leq (\varepsilon^2 \lg N) w(P(Q))$,
and, thus, if $C_i$ were from $\C'(Q)$ then it would be light and it would have
been deleted from $\OPT'(Q)$ previously. Hence, $C_i$ must be from $\C(Q_r)$.
Next, we introduce one subsegment for each interval of the form $(z,z')$ with
$z,z'\in \Zfine$ such that (a') no non-important interval from $\OPT''(Q,S)$
intersects $(z,z')$, (b') for each important interval $C_{i}\in\OPT''(Q,S)$ it
holds that $C_{i}\subseteq(z,z')$ or $C_{i}\cap(z,z')=\emptyset$ and (c') among
all pairs $(z,z')$ with properties (a') and (b'), $z' - z$ is minimal; these are
the subsegments $\T_{\sparse}$ mentioned in the first point of
Definition~\ref{def:structured-solution}. To see that these subsegments can only
contain $1/\varepsilon$ intervals from $\OPT''(Q,S)$, recall that each important
interval overlaps with a point in $\Zfine$ and in the previous paragraph we
removed every $1/\varepsilon$'th important interval. Thus, since we picked the
intervals $(z,z')$ with minimal size, each of them can contain at most
$1/\varepsilon$ intervals.

(5): To see the correctness of the inequality claimed in the lemma, note that each
interval from $\OPT''(Q,S)$ is contained in one of the previously constructed
subsegments in $\T$. Furthermore, observe that throughout this proof
we have only charged important intervals $C_i$ and their respective charge was
at most $\kappa' \varepsilon w_i$ for some absolute constant $\kappa' > 1$; we
did not charge any non-important intervals. Also notice that all important
intervals are contained in the subsegments in $\T_{\sparse}$ and all
non-important intervals are contained in the subsegments in $\T_{\dense}$.
This implies that
\begin{align*}
	&\sum_{T_{j}\in\T_{\sparse}}w(\OPT_{T_{j}}''(Q,S))+\sum_{T_{j}\in\T_{\dense}}w(\OPT_{T_{j}}''(Q,S)) \\
	\ge& w(\OPT^*(Q,S)) - \kappa'\cdot \epsilon \sum_{T_{j}\in\T_{\sparse}}w(\OPT_{T_{j}}''(Q,S)).
\end{align*}
Rearranging terms yields the inequality which we claimed in the lemma.

\section{Dynamic Independent Set of Rectangles and Hyperrectangles}
\label{sec:rectangles}
We study the setting where each input element $C_{i}\in\C$ is an axis-parallel $d$-dimensional
hyperrectangle, rather than a hypercube. Like for hypercubes, we assume
that each input hyperrectangle $C_{i}$ is defined by coordinates
$x_{i}^{(d')},y_{i}^{(d')}\in[0,N]$ where $|y_{i}^{(d')}-x_{i}^{(d')}|\ge1$
for each dimension $d'\in[d]$. We give an $O(\log^{d-1}N)$-approximation
algorithm with worst-case update time $d (1/\varepsilon)^{O(1/\varepsilon)} \log(n+N)^{O(1)}$.
Our strategy is to reduce
this case to the one-dimensional setting by paying a factor $\log^{d-1}N$,
similarly as in~\cite{agarwal98label}.

We classify the hyperrectangles into $\log N$ different classes.
We say that a hyperrectangle $C_{i}$ is of class 1 if $N/2\in[x_{i}^{(1)},y_{i}^{(1)})$.
Recursively, a hyperrectangle $C_{i}$ is of class $c$ if there a
$k\in\N$ such that $k\cdot N/2^{c}\in[x_{i}^{(1)},y_{i}^{(1)})$
but $C_{i}$ is not of class $c-1$. For each $c$ denote by $\C^{(c)}$
the rectangles of class $c$. 
\begin{lem}
For each hyperrectangle $C_{i}$ there is a class $c\in\{1,\dots,\log N\}$
such that $C_{i}\in\C^{(c)}$. 
\end{lem}

\begin{proof}
This holds since for each hyperrectangle $C_{i}$ and each dimension
$d'\in[d]$ and there is a $k\in\N$ such that $k=k\cdot N/2^{\log N}\in[x_{i}^{(1)},y_{i}^{(1)})$,
since we assumed that $|y_{i}^{(d')}-x_{i}^{(d')}|\ge1$ for each
hyperrectangle $C_{i}$. 
\end{proof}
For each class $\C^{(c)}$ with $c\in\{1,\dots,\log N\}$ we maintain
an independent set $A^{(c)}\subseteq\C^{(c)}$ and then output the
solution with maximum weight among the solutions $A^{(c)}$. We will
show that this reduction to the problem for one single class $\C^{(c)}$
loses only a factor of $\log N$. The key observation is that for
each class $c$ we can partition the hyperrectangles in $\C^{(c)}$
into at most $N$ different groups $\C^{(c,k)}$ such that no two
hyperrectangles in different groups overlap and the problem for each
group is equivalent to an instance of $(d-1)$-dimensional hyperrectangles.
Let $c$ be a class. For each odd integer $k$ let $\C^{(c,k)}:=\{C_{i}\in\C^{(c)}|k\cdot N/2^{c}\in[x_{i}^{(1)},y_{i}^{(1)})\}$. 
\begin{lem}
Let $c$ be a class. For each $C_{i}\in\C^{(c)}$ there is an odd
$k\in[N]$ such that $C_{i}\in\C^{(c,k)}$. 
\end{lem}

\begin{proof}
By definition of the classes for each hyperrectangle $C_{i}\in\C^{(c)}$
there is a $k\in\N$ such that $k\cdot N/2^{c}\in[x_{i}^{(1)},y_{i}^{(1)})$.
If $k$ was even, i.e., $k=2k'$ for some $k'\in\N$ then $k'\cdot N/2^{c-1}\in[x_{i}^{(1)},y_{i}^{(1)})$
and hence $C_{i}\in\C^{(c-1)}$ and therefore $C_{i}\notin\C^{(c)}$. 
\end{proof}
For each class $c$, we show that the sets $\C^{(c,k)}$ form independent
subinstances, i.e., there is no hyperrectangle in $\C^{(c,k)}$ that
overlaps with a hyperrectangle in $\C^{(c,k')}$ if $k\ne k'$. 
\begin{lem}
Let $c$ be a class. Let $C_{i}\in\C^{(c,k)}$ and $C_{i'}\in\C^{(c,k')}$
such that $k\ne k'$. Then $C_{i}\cap C_{i'}=\emptyset$. 
\end{lem}

\begin{proof}
Assume w.l.o.g.~that $k<k'$ and suppose by contradiction that $C_{i}\cap C_{i'}\ne\emptyset$.
Then it must be that $(x_{i}^{(1)},y_{i}^{(1)})\cap(x_{i'}^{(1)},y_{i'}^{(1)})\ne\emptyset$.
Let $k''$ be an even number with $k<k''<k'$. Then $k''\cdot N/2^{c}\in(x_{i}^{(1)},y_{i}^{(1)})$
or $k''\cdot N/2^{c}\in(x_{i'}^{(1)},y_{i'}^{(1)})$. Assume that
$k''\cdot N/2^{c}\in(x_{i}^{(1)},y_{i}^{(1)})$. Then $k''\cdot N/2^{c}=k''/2\cdot N/2^{c-1}\in(x_{i}^{(1)},y_{i}^{(1)})$
and hence $C_{i}\in\C^{(c-1)}$ which is a contradiction. 
\end{proof}
Let $c$ be a class. We maintain a solution $A^{(c)}\subseteq\C^{(c)}$.
Also, we maintain a search tree $S^{(c)}$ of depth $\log N$ for
the sets $\C^{(c,k)}$ with $\C^{(c,k)}\ne\emptyset$. For each
such set $\C^{(c,k)}$ there is a corresponding node $v_{c,k}$ in $S^{(c)}$.
In $v_{c,k}$ we store a data structures that maintain $\C^{(c,k)}$,
and a data structure that maintains an independent set of $(d-1)$-dimensional
hyperrectangles from an input set $\bar{\C}^{(c,k)}$ (that correspond
to $\C^{(c,k)}$). When a rectangle $C_{i}$ is inserted or removed
then we compute $c$ and $k$ such that $C_{i}\in\C^{(c,k)}$. If
$S^{(c)}$ does not contain a node $v_{c,k}$ for $\C^{(c,k)}$ then
we add such a node $v_{c,k}$ in time $O(\log N)$. Then, we add/remove
$C_{i}$ to/from $\C^{(c,k)}$, and we add/remove in $\bar{\C}^{(c,k)}$
the $(d-1)$-dimensional hyperrectangle $\tilde{C}_{i}:=(x_{i}^{(2)},y_{i}^{(2)})\times\dots\times(x_{i}^{(d)},y_{i}^{(d)})$
with weight $w_{i}$. Let $A^{(c,k)}\subseteq\bar{\C}^{(c,k)}$ denote
the new independent set for the input set $\bar{\C}^{(c,k)}$. We
update the solution $A^{(c)}$ to be $A^{(c)}:=\bigcup_{k:\C^{(c,k)}\ne\emptyset}\{C_{i}|\tilde{C}_{i}\in A^{(c,k)}\}$.

Let $c^{*}$ be the class $c$ with the solution $A^{(c)}$ of maximum
weight, i.e., such that $w(A^{(c^{*})})\ge w(A^{(c)})$ for each class
$c$. We output $A^{(c^{*})}$. The construction above yields the
following lemma. 
\begin{lem}
\label{lem:inductive-step}Given a data structure that maintains an
$\alpha$-approximate independent set of $(d-1)$-dimensional hyperrectangles
with update time $T$. Then there is a data structure that maintains
a $(\alpha\log N)$-approximate independent set of $d$-dimensional
hyperrectangles with update time $T+O(\log N)$. 
\end{lem}

A $(1+\varepsilon)$-approximate data structure for intervals (i.e., for $d=1$) is given by
Theorem~\ref{thm:weighted-intervals-1}. By induction we can prove the following
theorem, where the inductive step is given by Lemma~\ref{lem:inductive-step}. 
\begin{thm}
	There is a fully dynamic algorithm for the weighted maximum independent set of
	$d$-dimensional hyperrectangles problem that maintains a
	$(1+\epsilon)\log^{d-1}N$-approximate
	solution in worst-case update time~$d \updateweightedintervals$.
\end{thm}

\section{Lower Bounds}
\label{sec:lower-bounds}

Our algorithms above have worst-case update times of $\log^{O(d)}(N+n)$
and achieve approximation ratios of $2^{d}(1+\epsilon)$ and $O(2^{d})$,
respectively. Next, we show that one cannot improve these approximation ratios to
$1+\epsilon$ for any $d\ge2$, even if the hypercubes are unweighted,
if we allow amortized update time of $n^{O((1/\epsilon)^{1-\delta})}$
for some $\delta>0$, and if we restrict ourselves to the insertion-only
model, rather than the fully dynamic model. The reason is that then
we could use this algorithm to compute a $(1+\epsilon)$-approximate
solution for unweighted squares offline in time $n^{O((1/\epsilon)^{1-\delta})}$.
This is impossible, unless the Exponential Time Hypothesis (ETH) fails~\cite{marx2007optimality}. 
\begin{thm}
\label{thm:lower-bound}For any $d\ge2$ there is no dynamic algorithm
in the insertion-only model for unweighted independent set of $d$-dimensional
hypercubes with an approximation ratio of $1+\epsilon$ and an amortized
update time of $n^{O((1/\epsilon)^{1-\delta})}$ for any $\delta>0$,
unless the ETH fails.
\end{thm}
\begin{proof}
If we had such an algorithm for $d=2$ we could use it in order to
solve maximum weight independent set of (unweighted) squares by simply
inserting all $n$ squares in the input one by one which would need
$n^{O((1/\epsilon)^{1-\delta})}$ time in total. However, there is
no PTAS for independent set of squares with running time $n^{O((1/\epsilon)^{1-\delta})}$,
unless the ETH fails~\cite{marx2007optimality}. If we had such an
algorithm for any $d>2$ then, given an instance of maximum independent
set of squares, we could construct an equivalent instance of maximum
independent set of $d$-dimensional hypercubes, where for each dimension
$d'\in\{3,\dots,d\}$ we define that $x_{i}^{(d')}=0$ and $y_{i}^{(d')}=1$
for each hypercube $C_{i}$. Then by the same reasoning we would get
a PTAS for independent set of unweighted squares with a running time
of $n^{O((1/\epsilon)^{1-\delta})}$, contradicting~\cite{marx2007optimality}.
\end{proof}

We also provide a cell-probe lower bound showing that any algorithm solving the
(exact) dynamic maximum weighted independent set must have an amortized update
and query time $\Omega(\lg N/\lg \lg N$). Note that the theorem holds for the
version of the problem where a query returns the weight of independent set (and
not the explicit independent set).
\begin{thm}
	Any algorithm for the (exact) dynamic weighted maximum weighted independent set of
	intervals problem, where the intervals are contained in the space $[0,N]$ and
	have size at least 1, the amortized update and query time is
   	$\Omega(\log N/\log \log N)$.
\end{thm}
\begin{proof}
	We show the theorem by a reduction from the dynamic prefix sum problem.  In
	the dynamic prefix sum problem we are given an array $A$ of size $N$, which
	is initialized by zeros, and there are two operations:
	(1) Update($i,\Delta$): Set $A[i]=\Delta$, where $\Delta \in \{0,1\}$.
	(2) Query($i$): Return $(\sum_{1 \le j \le i} A[j]) \mod 2$.

	Fredman and Saks~\cite{FredmanS89} showed that there exists a sequence of
	$O(N)$ updates and $O(N)$ queries that takes time $\Omega(N \log N/\log \log N)$
	to process in the cell probe model with word size $\omega = \Omega(\lg N)$.
	We now show how to implement the dynamic partial sum problem using a dynamic
	(exact) weighted independent set data structure such that each operation in
	the dynamic partial sum problem requires a constant number of operations in
	the independent set data structure. This proves the theorem.

	Initially there is no interval. To implement an Update($i,\Delta$)
	operation we do the following. If $\Delta = 0$, remove any interval
	$[i-1,i)$ if there exists one; if $\Delta = 1$, insert the interval
	$[i-1,i)$. To implement a Query($i$) operation, insert an interval $[i,N]$
	of weight $N$ and ask an query in the resulting set of intervals. As none of
	the intervals of size 1 overlap, the answer $k$ will be of value $N+x$,
	where $x$ equals $\sum_{1 \le j \le i} A[j]$. Returning $x \mod 2$ provides
	a correct answer to the partial sum query. Before finishing the query
	procedure, we delete the previously inserted interval $[i,N]$ with
	weight~$N$.
\end{proof}

\section{Omitted Proofs}
\label{sec:omitted}

\subsection{Proof of Lemma~\ref{lem:offset}}

To construct $\OPT'$, the intuition is that we first delete some
of the hypercubes from $\OPT$ while losing only a total weight of
$\epsilon\cdot w(\OPT)$. Then we group the remaining hypercubes into
groups according to their sizes such that the sizes of two hypercubes
in different groups differ by a factor $1/\epsilon^{2}$ and within
each group the sizes differ by a factor $1/\epsilon^{O(1/\epsilon)}$.
Then we define the set $\off(\epsilon)$ of size $O((d/\epsilon)^{1/\epsilon})$
and prove that if we draw the grid offset $a$ from $\off(\epsilon)$
uniformly at random that only a small fraction of the remaining hypercubes
need to be deleted.

Let $K:=\log2d/\epsilon$. Let $a'\in\{0,\dots,1/\epsilon-1\}$ be an
offset to be defined later. We delete each hypercube $C_{i}\in\OPT$
such that $\ell(C_{i})\in\{a'K+k\cdot K/\epsilon+b|k\in\Z,b\in\{0,\dots,K-1\}\}$.
Let $\OPT''$ denote the remaining hypercubes. Note that for each
hypercube $C_{i}\in\OPT$ there exists only one offset $a'\in\{0,\dots,1/\epsilon-1\}$
such that $C_{i}$ is deleted. Hence, there is an offset $a'\in\{0,\dots,1/\epsilon-1\}$
such that $w(\OPT'')\ge(1-\epsilon)w(\OPT)$.

We group the hypercubes in $\OPT''$ into groups $\OPT''_{j}$ where
for each $j\in\Z$ we define $\OPT''_{j}$ to be all hypercubes with
levels in the range $a'K+j\cdot K/\epsilon,\dots,a'K+(j+1)\cdot K/\epsilon-K+1$,
i.e., $\OPT''_{j}:=\{C_{i}|\ell(C_{i})\in\{a'K+j\cdot K/\epsilon+b|b\in\{0,\dots,K/\epsilon-K+1\}\}\}$.
Observe that if $C_{i}\in\OPT''_{j}$ and $C_{i'}\in\OPT''_{j'}$
for $j<j'$ then $\ell(C_{i})\le\ell(C_{i'})-K+1$.

Imagine first that we had only one non-empty group $\OPT''_{j}$.
Note that for each hypercube $C_{i}\in\OPT''_{j}$ we have that $s_{i}\ge\epsilon N/(d2^{\ell(C_{i})-1})\ge\epsilon N/(d2^{a'K+(j+1)\cdot K/\epsilon-K+1})=:\delta_{j}$.
Also, if there is a grid cell $Q$ with $C_{i}\in\C'(Q)$ such that
$\ell(Q)=\ell(C_{i})$ then $Q$ has size $N/2^{\ell(C_{i})-1}=N/2^{\ell(Q)-1}\le N/2^{a'K+j\cdot K/\epsilon}=:\Delta_{j}$
in each dimension. Note that $\Delta_{j}/\delta_{j}=d2^{K/\epsilon-K+1}/\epsilon$.
Suppose that we draw a value $r$ uniformly at random from $\{0,\dots,d2^{K/\epsilon-K+1}/\epsilon-1\}$
and then use the offset $a'_{r}:=r\cdot\delta_{j}$. We say that a
hypercube $C_{i}\in\OPT''_{j}$ \emph{survives} if $C_{i}\in\C'(Q)$
for some cell $Q$. We claim that the probability that $C_{i}$ survives
is at least $1-O(\epsilon)$. To see this, note that $C_{i}$ survives
if and only if for no dimension $d'$ there is a value $k\in\Z$ such
that $a'_{r}+k\cdot N/2^{\ell(C_{i})-1}\in(x_{i}^{(d')},y_{i}^{(d')})$.
Note that $N/2^{\ell(C_{i})-1}$ is a multiple of $\delta_{j}$ and
$s_{i}<2\epsilon N/(d2^{\ell(C_{i})-1})$. Hence, for one dimension
$d'\in[d]$ there are at most $O(\epsilon)\cdot2^{K/\epsilon-K+1}/\epsilon$
values for $r\in\{0,\dots,d2^{K/\epsilon-K+1}/\epsilon-1\}$ such
that $a'_{r}+k\cdot N/2^{\ell(C_{i})-1}\in(x_{i}^{(d')},y_{i}^{(d')})$.
Therefore, there are at most $O(\epsilon d)\cdot2^{K/\epsilon-K+1}/\epsilon$
values for $r\in\{0,\dots,d2^{K/\epsilon-K+1}/\epsilon-1\}$ such
that $C_{i}$ does not survive.

Suppose that we draw $r$ uniformly at random from $\{0,\dots,d2^{K/\epsilon-K+1}/\epsilon-1\}$
and define the offset $a_{r}:=r\cdot\sum_{\hat{j}=-1}^{\log N}\delta_{\hat{j}}$.
We claim that still each $C_{i}\in\OPT'$ survives with probability
at least $1-O(\epsilon)$. Let $C_{i}\in\OPT''_{j}$ for some $j$.
Note that if $j'<j$ then $\delta_{j'}$ is a multiple of $N/2^{\ell(C_{i})-1}$
since

\[
\begin{aligned}\frac{\delta_{j'}}{N/2^{\ell(C_{i})-1}} & =\frac{\epsilon N/(d2^{a'K+(j'+1)\cdot K/\epsilon-K+1})}{N/2^{\ell(C_{i})-1}}\\
 & =2\epsilon\frac{2^{\ell(C_{i})-1}}{d2^{a'K+(j'+1)\cdot K/\epsilon-K+1}}\\
 & \ge2\epsilon\frac{2^{a'K+j\cdot K/\epsilon-1}}{d2^{a'K+(j'+1)\cdot K/\epsilon-K+1}}\\
 & \ge2\epsilon2^{K-2}/d\\
 & =1
\end{aligned}
\]
where we use that $1/\epsilon$ is a power of 2, $\ell(C_{i})\ge a'K+j\cdot K/\epsilon$,
and $2^{K}=2d/\epsilon$. Hence, $C_{i}$ survives for the offset
$a_{r}$ if and only if $C_{i}$ survives for the offset $r\cdot\sum_{\hat{j}=j}^{\log N}\delta_{\hat{j}}$.
On the other hand, note that

\[
\begin{alignedat}{1}\sum_{\hat{j}=j+1}^{\log N}\delta_{\hat{j}} & \le\delta_{j}\cdot\sum_{\hat{j}=j+1}^{\log N}\frac{1}{2^{(\hat{j}-j)\cdot K/\epsilon}}\\
 & \le\delta_{j}\cdot\sum_{k=1}^{\infty}\frac{1}{2^{k\cdot K/\epsilon}}\\
 & \le\delta_{j}\cdot\sum_{k=1}^{\infty}\left(\frac{\epsilon}{2d}\right)^{k/\epsilon}\\
 & =O(\epsilon/d)\delta_{j}\\
 & \le O(\epsilon/d)s_{i}.
\end{alignedat}
\]
Thus, for each dimension $d'\in[d]$ there are again at most most
$O(\epsilon)\cdot2^{K/\epsilon-K+1}/\epsilon$ values for $r\in\{0,\dots,d2^{K/\epsilon-K+1}/\epsilon-1\}$
such that $a'_{r}+k\cdot N/2^{\ell(C_{i})-1}\in(x_{i}^{(d')},y_{i}^{(d')})$.
Therefore, we define $\off(\epsilon):=\{a_{r}|r\in\{0,\dots,d2^{K/\epsilon-K+1}/\epsilon-1\}\}$
and observe that if we draw $r$ uniformly at random from $\off(\epsilon)$
then $C_{i}$ survives with probability at least $1-O(\epsilon)$.
Denote by $\OPT'$ the resulting solution.%

\subsection{Proof of Lemma~\ref{lem:properties-grid}}
	The first claim follows immediately from the definition of the $Q_{\ell,k}$.
	The second claim follows from the fact that the volume of a hypercube
	$C$ from level $\ell(Q)$ is $(\varepsilon N/2^{\ell(Q)-1})^{d}$;
	dividing the volume of $Q$ by the minimum volume of hypercubes in
	$\C'(Q)$ proves the claim. The third claim follows from the fact
	that for each level $\ell$, the intersection of any two cells of
	$\G_{\ell}$ has zero volume. The fourth claim holds since the hierarchy has depth $\log N$ and
	the grid cells at each level are disjoint. 

\subsection{Proof of Lemma~\ref{lem:data-structure}}
	We use a data structure for range counting/reporting as a subroutine.
	For range counting/reporting one seeks to construct a data structure
	that obtains as input a set of points $x_{1},\dots,x_{n}\in\R^{d'}$
	with weights $w_{1},\dots,w_{n}$ and the data structure and supports
	the following operations: 
	\begin{itemize}
		\item inserting/deleting points 
		\item answering \emph{range reporting queries}: Given a hyperrectangle $A=[a_{1},b_{1}]\times\dots\times[a_{d},b_{d}]$,
		report all points $x_{i}$ with $x_{i}\in A$ 
		\item answering \emph{range counting queries}: Given a hyperrectangle $A$, return
		the sum of the weights of all points in $A$, i.e., return $\sum_{i:x_{i}\in A}w_{i}$ 
	\end{itemize}
	As pointed out in Lee and Preparata~\cite{lee84computational}, one
	can combine the algorithms of Willard and Lueker~\cite{willard85adding}
	and Edelsbrunner~\cite{edelsbrunner1981note} to obtain a data structure
	for the range counting/reporting problem with the following guarantees.
	Points can be inserted and deleted in worst-case update time $O(\lg^{d'}n)$.
	Range counting queries can be answered in worst-case query time $O(\lg^{d'-1}n)$
	and range reporting queries can be answered in worst-case query time
	$O(\lg^{d'-1}n+F)$, where $F$ is the number of points which shall
	be reported. Also, the data structure can be initialized with an empty
	set of points in time $O(1)$.

	To construct the data structure that is claimed by the lemma, we use two range
	counting/reporting data structures, one with $d':=2d$ and
	another one with $d' = 2d+1$. When our data structure is
	initialized, the set of hypercubes $\C'$ is empty and we initialize both
	the range reporting data structures with an empty set of points in time
	$O(1)$. Also, we initialize a counter for $|\C'|$ that we update
	in worst-case update time $O(1)$ whenever a hypercube is added or
	deleted. In this way, we can report whether $\C'=\emptyset$ in time
	$O(1)$.

	Now suppose a hypercube
	$C=(x_{C}^{(1)},y_{C}^{(1)})\times\dots\times(x_{C}^{(d)},y_{C}^{(d)})$ of size
	$s$ inserted or deleted from the data structure. Then we insert/delete
	the point $C'=(x_{C}^{(1)},y_{C}^{(1)},\dots,x_{C}^{(d)},y_{C}^{(d)})$
	into the range reporting data structure with $d' = 2d$ and we insert/delete the point
	$C''=(x_{C}^{(1)},y_{C}^{(1)},\dots,x_{C}^{(d)},y_{C}^{(d)},s)$ into the data
	structure with $d'=2d+1$ (recall that $s = y_C^{(1)} - x_C^{(1)}$); both
	points are inserted with weight $w(C)$.
	Both of these updates can be done in worst-case update time $O(\lg^{2d+1}|\C'|)$.

	Now suppose the data structure obtains a hyperrectangle $B=(x_{B}^{(1)},y_{B}^{(1)})\times\dots\times(x_{B}^{(d)},y_{B}^{(d)})$
	and the data structure needs to report a hypercube $C_i \in \C'$ with
	$C_i \subseteq B$.
	Then we query the range counting data structure with $d' = 2d$ for the range 
	\begin{align}
	\label{eq:range}
		B'=\prod_{j=1}^{d}(x_{B}^{(j)},y_{B}^{(j)})\times(x_{B}^{(j)},y_{B}^{(j)}).
	\end{align}
	Given a hypercube $C_{i}=(x_{i}^{(1)},y_{i}^{(1)})\times\dots\times(x_{i}^{(d)},y_{i}^{(d)})$
	then $C_{i}\subseteq B$ \emph{if and only if} its corresponding point
	$p_{i}=(x_{i}^{(1)},y_{i}^{(1)},\dots,x_{i}^{(d)},y_{i}^{(d)})$ satisfies
	$p_{i}\in B'$. This follows from the fact that $(x_{C}^{(j)},y_{C}^{(j)})\subseteq(x_{B}^{(j)},y_{B}^{(j)})$
	if and only if $x_{B}^{(j)}\leq x_{C}^{(j)}\leq y_{B}^{(j)}$ and
	$x_{B}^{(j)}\leq y_{C}^{(j)}\leq y_{B}^{(j)}$ for each $j\in[d]$.
	Since $p_{i}\in B'$ asserts that the above property holds for all
	$j\in[d]$, this implies that $p_{i}\in B'$ iff $C_{i}\subseteq B$.
	This implies the correctness of the query operation which can be performed
	with worst-case query time $O(\lg^{2d-1}|\C'|)$. 

	Now suppose the data structure obtains a hyperrectangle $B$ as above and is
	asked to return the smallest hypercube $C_i \in \C'$ with $C_i \subseteq B$.
	For this purpose, the data structure also maintains a list $L$ of all
	hypercubes in $\C'$ ordered by their sizes; note that $L$ can be maintained
	with worst-case update time $O(\lg |\C'|)$ whenever a hypercube is inserted into
	or removed from $\C'$ and that this time is subsumed by the update time for
	the range reporting data structures. To answer the query, the algorithm uses
	binary search on $L'$ as follows. It initializes $L' = L$. It sets $k$ to
	the median index of $L'$ and sets $s_k$ to the size of the hypercube
	$L'(k)$, where $L'(k)$ is the $k$-th hypercube in $L'$.
	Now the algorithm uses the range reporting data structure for $d' =
	2d+1$ to check if there exists a hypercube $C_i$ which is contained in the
	range $B' \times [0,s_k]$, where $B'$ is as defined in
	Equation~\eqref{eq:range}. This is done exactly as defined in the previous
	paragraph but with the minor modification that now we have one additional
	dimension which ensures that the size of any returned hypercube is in the
	range $[0,s_k]$.  If such a hypercube exists, the binary search recurses on
	the sublist $L' = L' \cap \{L(0),\dots,L(k)\}$ and, otherwise, it
	recurses on the sublist $L' = L' \cap \{L(k),\dots,L(|L|)\}$.  After $O(\lg n)$
	iterations the algorithm has found a hypercube with the desired properties
	(if one exists).  Since the binary search makes $O(\lg n)$ queries and each
	of these queries requires time $O(\lg^{2d} n)$, finding the desired
	hypercube takes worst-case time $O(\lg^{2d+1}n)$. Also, if no such hypercube
	exists, the algorithm will assert this.

	Now we prove the fourth point of the lemma. We use almost the same procedure
	as in the previous paragraph. We assume that the data structure also
	maintains an ordered list $\L$ of all hypercubes
	$C_i=(x_i^{(1)},y_i^{(1)})\in\C'$ ordered by their $y^{(1)}$-coordinates;
	as argued before, maintaining $\L$ does not increase the update time of the
	data structure.  To answer the query, the algorithm finds the
	smallest $y^{(1)}$-coordinate in $\L$ with $y_i^{(1)}\geq t$; suppose the
	corresponding hypercube has index $j$ in $\L$.  Now the algorithm uses
	binary search on the sublist $\L'=\{\L(j),\dots,\L(|\L|)\}$ as follows. Let
	$k$ be the median index of $\L'$ and let $y_k^{(1)}$ be the
	$y^{(1)}$-coordinate of $\L'(k)$. Now the algorithm queries the range
	reporting data structure for $d' = 2$ if there exists a hypercube in $\C'$
	in the range $[t,y_k^{(1)})$.  If this is the case, the binary search
	recurses on $\L' = \L' \cap \{\L(j),\dots,\L(k)\}$ and, otherwise, it
	recurses on the sublist $\L' = \L' \cap \{\L(k),\dots,\L(|\L|)\}$.  Clearly,
	after $O(\lg n)$ iterations the algorithm has found a hypercube with the
	desired properties.  Since the binary search makes $O(\lg n)$ queries and
	each of these queries requires time $O(\lg n)$ by the third point of the
	lemma, finding the desired hypercube takes time $O(\lg^{2}n)$. Also, if no
	such hypercube exists, the algorithm will assert this. 

\subsection{Proof of Lemma~\ref{lem:range-counting-data-structure}}
	The claims of the first three points follow immediately from the results
	reported in~\cite{lee84computational,willard85adding,edelsbrunner1981note}.
	The fourth claim be achieved in as follows. For each dimension $d'\in[d]$,
	we maintain a separate binary search tree of points $P_{d'}$ in which the
	points ordered by their $d'$-th coordinate. Whenever a point is deleted or
	inserted into $P$, we update all of theses search trees $P_{d'}$
	accordingly. When we obtain a query as presented in the fourth point of the
	lemma, then we use binary search to find the median of the $d'$-th
	coordinates in the range $[x,z]$. Note that this value $y$ satisfies the
	claim because the interval $[x,z]$ is closed and the intervals $(x,y)$ and
	$(y,z)$ are open.

\subsection{Proof of Lemma~\ref{lem:construct-aux-grid}}
	For each dimension $d'\in[d]$ we do the following. First, we define
	$z_{1}^{(d')}:=x_{Q}^{(d')}$.
	Then for $j=1,2,\dots$, we perform binary
	search in order to find the largest value $z\le y_{Q}^{(d')}$ such
	that the total weight of the points in
	$\left(\R^{d'-1}\times(z_{j}^{(d')},z)\times\R^{d-d'}\right)\cap P(Q)$
	is at most $\epsilon^{d+2}\tilde{W}/(d^{d+1} \log N)$. This is done as follows.
	Using Property~4 of Lemma~\ref{lem:range-counting-data-structure} on the
	interval $(z_j^{(d')},y_Q^{(d')})$, we find a candidate
	value $y$ for $z$. Then we use Property~3 of Lemma~\ref{lem:range-counting-data-structure}
	to query the weight of the points in
	$\left(\R^{d'-1}\times(z_{j}^{(d')},y)\times\R^{d-d'}\right)\cap P(Q)$.
	Depending on whether this is less than or more than
	$\epsilon^{d+2}\tilde{W}/(d^{d+1}\log N)$, we recurse on the intervals 
	$(y,y_Q^{(d')})$ or $(z_j^{(d')},y)$, respectively.
	When the recursion stops, we set $z_{j+1}^{(d')}:=z$.
	Note that this requires $O(\lg |P(Q)|)$ recursion steps and each recursion step
	takes time $O(\lg^{d-1} |P(Q)| + \lg |P(Q)|) = O(\lg^{d-1} |P(Q)|)$.

	Note that the third property of the lemma holds since for each $j=2,3,\dots$, we
	picked the maximum $z$ such that $(z_j^{(d')},z)$ has weight at most 
	$\epsilon^{d+2}\tilde{W}/(d^{d+1} \log N)$. Hence, there can only be
	$d^{d+1} \lg N / (\epsilon^{d+2}\tilde{W})$ iterations for $j$.

	Doing this for all dimensions $d'\in[d]$ takes
	total time $O(d\cdot\log |P(Q)|\cdot\log^{d-1}|P(Q)|\cdot d^{d+1}\log N/\epsilon^{d+2})$.

\subsection{Proof of Lemma~\ref{lem:iterations}}
	The first claim holds if we can show that for each of the $\lg W$ weight classes
	$\C_{k}$, we can add at most $(d/\epsilon)^{d}$ hypercubes to $\bar{\C}(Q)$ in all
	iterations of the algorithm. Indeed, this is similar to the statement of
	Lemma~\ref{lem:properties-grid}. Note, however, that we cannot apply the lemma
	directly because the algorithm does not check if two of the added hypercubes from the same
	weight class intersect or not.
	However, when we add a hypercube $C_{i} \in \C_k'$ to $\bar{\C}(Q)$ then we added its
	vertices to $P(Q)$ with weight $w_{i}$ and we did not remove these points in case
	that we removed $C_{i}$ from $\bar{\C}(Q)$ later on.  Now consider the case when we select
	another hypercube $C_{i'} \in \C_k'$ and suppose that $C_i$ and $C_{i'}$ intersect.
	Then for the respective set $A$ for which the query returned $C_{i'}$, it cannot
	hold that $(1+\epsilon)^{k}\ge2w(P(Q)\cap A)$ (since one vertex of $C_i$ must be
	in $A$ and because $C_i$ and $C_{i'}$ are from $\C_k(Q)$ and thus their weights
	can only differ by a
	$(1+\varepsilon)$-factor). Thus, the query cannot have returned $C_{i'}$.
	Thus for each $k$, all hypercubes which are picked from $\C_k$ do
	not intersect. Altogether, by Lemma~\ref{lem:properties-grid}
	there are at most $\left(\frac{d}{\epsilon}\right)^{d} \log W$ iterations in which we select
	an additional hypercube. 

	Initializing the auxiliary grid due to Lemma~\ref{lem:construct-aux-grid}
	takes time
	$O\left(\left(\frac{d}{\varepsilon}\right)^{d+2} \cdot\log^{d}n\cdot(\log N)\right)$.
	Hence, we need total time 
	$$O\left(
		\left(\frac{d}{\varepsilon}\right)^{d+2}\cdot\log^{d}n\cdot(\log N)
		+((\log_{1+\epsilon}W)/\epsilon^{d}+1)\left(\left(\frac{d}{\epsilon}\right)^{d+2}(\log N) +1\right)^{2d}\cdot\log^{2d-1}n
	\right)$$
	to compute $\bar{\C}(Q)$.%

\end{document}